\documentclass{jgaa-artTwo}
\usepackage{microtype}
\usepackage{xspace}
\usepackage{subcaption}
\usepackage[utf8]{inputenc}
\usepackage{booktabs}
\PassOptionsToPackage{dvipsnames}{xcolor}
\DeclareGraphicsExtensions{.pdf,.png,.eps}
\graphicspath{{figs/}}
\definecolor{dark blue}{rgb}{0.121,0.47,0.705}
\definecolor{dark red}{rgb}{0.89,0.102,0.109}
\definecolor{dark green}{rgb}{0.2,0.627,0.172}
\definecolor{dark orange}{rgb}{1,0.498,0}
\definecolor{dark purple}{rgb}{0.415,0.239,0.603}
\definecolor{dark pink}{rgb}{0.969, 0.506, 0.749}
\newcommand{\bl}{\color{dark blue}}
\newcommand{\re}{\color{dark red}}
\newcommand{\gr}{\color{dark green}}
\newcommand{\og}{\color{dark orange}}
\newcommand{\pu}{\color{dark purple}}
\newcommand{\pk}{\color{dark pink}}
\newcommand{\gy}{\color{darkgray}}
\usepackage[inline]{enumitem}

\definecolor{defblue}{rgb}{0.08235294118,0.3098039216,0.537254902}
\let\emph\relax
\DeclareTextFontCommand{\emph}{\color{defblue}\em}

\usepackage{mathtools}
\usepackage{amssymb}
\DeclarePairedDelimiter\set{\{}{\}}
\DeclarePairedDelimiter\abs{\lvert}{\rvert}
\DeclarePairedDelimiter\ceil{\lceil}{\rceil}

\DeclarePairedDelimiter\croc{\langle}{\rangle}
\def\Oh{\ensuremath{\mathcal{O}}}
\def\Sgiad{\ensuremath{\nwarrow}\xspace}
\def\Svert{\ensuremath{\uparrow}\xspace}
\def\Sdiag{\ensuremath{\nearrow}\xspace}
\def\Shori{\ensuremath{\rightarrow}\xspace}
\newcommand{\cT}{{\mathcal{T}}}
\newcommand{\gsquare}{$G_\square$\xspace}
\newcommand{\gsquares}{$G_\square$s\xspace}
\newcommand{\gshift}{$G_\leftrightarrow$\xspace}
\newcommand{\gshifts}{$G_\leftrightarrow$s\xspace}
\newcommand{\pms}{\textsc{Planar Monotone 3-SAT}\xspace}
\newcommand{\pmrs}{\textsc{Planar Monotone Rectilinear 3-SAT}\xspace}

\usepackage{url}
\usepackage{cite}
\PassOptionsToPackage{unicode}{hyperref}
\usepackage[capitalise,nameinlink]{cleveref}
\crefname{corollaryCounter}{Corollary}{Corollaries}

\newtheorem{observation}[theorem]{Observation}
\newtheorem{corollary}[theorem]{Corollary}

\newcommand{\etal}{{et~al.}\ }

\begin{document}

\doi{10.7155/jgaa.00xxx}
\Issue{0}{0}{0}{0}{0}
\HeadingAuthor{J.\ Klawitter and J.\ Zink}
\HeadingTitle{Upward Planar Drawings with Three and More Slopes}
\title{Upward Planar Drawings with Three and More Slopes}
\Ack{We thank the anonymous reviewers for their helpful comments.}
\authorOrcid[first]{Jonathan~Klawitter}{jo.\{lastname\}@gmail.com}{0000-0001-8917-5269}
\authorOrcid[first]{Johannes~Zink}{\{lastname\}@informatik.uni-wuerzburg.de}{0000-0002-7398-718X}
\affiliation[first]{Universität Würzburg, Würzburg, Germany}

\maketitle

\pdfbookmark[1]{Abstract}{Abstract}
\begin{abstract}
The \emph{slope number} of a graph $G$ is the smallest number of slopes needed 
for the segments representing the edges in any straight-line drawing of $G$.
It serves as a measure of the visual complexity of a graph drawing.
Several bounds on the slope number for particular graph classes have been established, 
both in the planar and the non-planar setting.
Moreover, the slope number can also be defined for directed graphs and upward planar drawings.

We study upward planar straight-line drawings that use only a constant number of slopes.
In particular, for a fixed number~$k$ of slopes,
we are interested in whether a given directed graph $G$ 
with maximum in- and outdegree at most $k$
admits an upward planar $k$-slope drawing.
We investigate this question both in the fixed and the
variable embedding scenario.
We show that this problem is in general NP-hard 
to decide for outerplanar graphs ($k = 3$) and planar graphs ($k \ge 3$).
On the positive side, we can decide whether a given cactus graph
admits an upward planar $k$-slope drawing and, in the affirmative, construct such a drawing 
in FPT time with parameter $k$.
Furthermore, we can determine the minimum number of slopes required for a given tree in linear time 
and compute the corresponding drawing efficiently.
\end{abstract}

\section{Introduction}
One of the main goals in graph drawing is to generate clear drawings.
Depending on the particular use case, we may request that a graph drawing has specific properties
and use quality measures for evaluation.
Classic examples are planarity or the number of crossings, the required drawing area, 
neighborhood preservation, or the stress of a drawing~\cite{DETT99}.
When we visualize directed graphs (digraphs for short) that model hierarchical relations,
we may represent edge directions explicitly by letting each edge point upward.
Combining this with planarity, we get upward planar drawings.

Schulz~\cite{Sch15} proposed the \emph{visual complexity} as a quality measure,
which is determined by the number of different geometric primitives used for the drawing.
Experiments have shown that people tend to prefer a low visual complexity~\cite{KMS18}.
Therefore, we may want to use few segments or arcs, or -- as in our case --
only few slopes for the edges.
Using a limited number of slopes is common for schematic drawings:
Various diagram types are orthogonal drawings and use only two different slopes.
With three or four slopes we get hexalinear and octilinear drawings, respectively,
which find applications in metro maps~\cite{NW11}, visualization of chemical cyclic compounds~\cite{FGR04}, and VLSI~\cite{MHS05}.

In this paper, we are interested in a combination of these requirements, namely,
we study upward planar straight-line drawings that use a fixed number of slopes;
\cref{fig:baseExample} shows an example with three slopes.
In particular, given a digraph $G$ and a fixed number~$k$ of slopes,
we ask whether $G$ admits an upward planar straight-line drawing with $k$ slopes.
Also, we have to distinguish whether $G$ comes with an embedding in the plane or not.
In the following, we define above drawing concepts formally, discuss related work,
and introduce the precise problem we consider.

\paragraph{Upward Planarity.}
An \emph{upward planar drawing} of a digraph $G$ is a planar drawing of $G$
where every edge $uv$ is drawn as a monotonic upward curve from $u$ to $v$. 
We call $G$ \emph{upward planar} if it admits an upward planar drawing and
\emph{upward plane} if it is equipped with an upward planar embedding. 
Note that an upward planar embedding, given by the edge order around each vertex,
is necessarily \emph{bimodal}, 
that is, each cyclic sequence of incident edges
can be split into two contiguous subsequences of incoming edges
and outgoing edges~\cite{DETT99}. 
Di Battista and Tamassia~\cite{DT88} have shown that if a digraph is upward planar, 
then it also admits an upward planar straight-line drawing.

When we are given a digraph $G$, we have to distinguish whether $G$ comes with an upward planar embedding or not.
In the latter case, we are required to find an upward planar embedding of $G$
(if one exists) in order to draw $G$ upward planar.
However, Garg and Tamassia~\cite{GT01} have shown 
that upward planarity testing is an NP-complete problem for general digraphs.
This even holds for digraphs with a bounded maximum degree $\Delta$, where $\Delta \geq 2$~\cite{KM21}.
On the positive side, there exist several fixed-parameter tractable (FPT) algorithms for general digraphs~\cite{Cha04,HL06,DGL10,CDFGRS22}
and polynomial-time algorithms for special graph classes such as single-source digraphs~\cite{BDMT98}, 
outerplanar digraphs~\cite{Pap94}, series-parallel digraphs~\cite{DGL10}, and triconnected digraphs~\cite{BDLM94}. 
If the embedding of a digraph is given, upward planarity can be tested in polynomial time~\cite{BDLM94}.

\paragraph{$k$-Slope Drawings.}
A \emph{$k$-slope drawing} of a (not necessarily directed) graph $G$ is a straight-line drawing of $G$
where every edge is drawn with one of at most $k$ different slopes; see \cref{fig:baseExample}.
Wade and Chu~\cite{WC94} defined the \emph{slope number} of~$G$ as the smallest $k$ 
such that $G$ admits a $k$-slope drawing. 
If only (upward) planar drawings are allowed, the number is called the \emph{(upward) planar slope number} of $G$.
Both the slope number and the planar slope have been studied extensively in the past,
mostly with the goal of finding upper bounds for particular graph classes~\cite{WC94,PP06,DSW07,DESW07,MS09,MP11,KPP13,LLMN13,JJKLTV13,DLM15,KMW14,DLM18,BKM19}.
We give a few examples. 
While Pach and P{\'a}lv{\"o}lgyi~\cite{PP06} have shown
that graphs with bounded degree $\Delta$ ($\Delta \geq 5$) can have arbitrarily large slope number,
Keszegh, Pach, P{\'a}lv{\"o}lgyi, and T{\'o}th~\cite{KPPT08} have shown that five slopes suffice for cubic graphs.
Later, Mukkamala and Szegedy~\cite{MS09} improved this to four slopes.
Restricted to planar drawings, 
Dujmović, Eppstein, Suderman and Wood~\cite{DESW07} showed, among other results,
that all plane cubic graphs admit planar 3-slope drawings.
In general, determining the planar slope number of a graph 
is hard in the existential theory of the reals ($\exists \mathbb{R}$), which was shown by Hoffmann~\cite{Hof17}.

Recently, the interest in upward planar drawings on few slopes has grown.
Bekos, Di Giacomo, Didimo, Liotta, and Montecchiani~\cite{BDDLM18} 
showed that every so-called bitonic $st$-graph $G$ with maximum degree $\Delta$
admits an upward planar 1-bend $\Delta$-slope drawing.
Complementarily, Di Giacomo, Liotta, and Montecchiani~\cite{DLM20} proved the analogous result for series-parallel digraphs.
We want to point out that both Bekos \etal and Di Giacomo \etal 
use horizontally drawn segments facing in both directions.
Brückner, Krisam, and Mchedlidze~\cite{BKM19} studied level-planar drawings with a fixed slope set,
that is, upward planar drawings where each vertex is drawn on a predefined integer y-coordinate (its level).
Older works include results by Czyzowicz, Pelc, and Rival~\cite{CPR90} and Czyzowicz~\cite{Czy91} on lattices
and several results for trees~\cite{CDP92,BBBHMU08,BM11,BM13}.

In a companion paper to this one,  
Klawitter and Mchedlidze~\cite{KM21} show that
it can be decided in linear time whether a given upward plane digraph admits an upward planar 2-slope drawing.
For the variable embedding scenario and two slopes,
they give a linear-time algorithm for single-source digraphs,
a quartic-time algorithm for series-parallel digraphs, and
an FPT algorithm for general digraphs.
Quapil~\cite{Qua21} recently also considered upward planar $k$-slope drawings.
His results include an extension of Hoffman's $\exists \mathbb{R}$-hardness 
from the planar slope number to the upward planar slope number, 
drawings of series-parallel graphs for $k = 3$, 
and area requirements of upward planar $k$-slope drawings for ordered trees, cacti, and series-parallel graphs. 

\begin{figure}[tb]
	\centering
	\begin{subfigure}[t]{0.20 \linewidth}
		\centering
		\includegraphics[page=3]{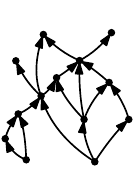}
		\caption{Digraph $G$.}
		\label{fig:baseExample:G}
	\end{subfigure}
	\begin{subfigure}[t]{0.25 \linewidth}
		\centering
		\includegraphics[page=6]{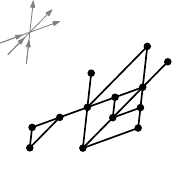}
		\caption{General setting.}
		\label{fig:baseExample:general}
	\end{subfigure} 	 
	\begin{subfigure}[t]{0.25 \linewidth}
		\centering
		\includegraphics[page=5]{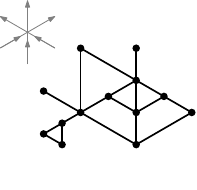}
		\caption{Uniform angles setting.}
		\label{fig:baseExample:uniform}
	\end{subfigure}
	\begin{subfigure}[t]{0.25 \linewidth}
		\centering
		\includegraphics[page=4]{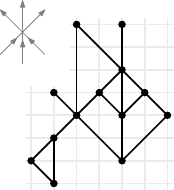}
		\caption{Regular grid setting.}
		\label{fig:baseExample:grid}
	\end{subfigure}
	\caption{The upward planar straight-line line drawings of the digraph $G$ realize different slope sets.
		In the remained of this paper, edges are drawn upward while arrow heads are omitted.}
	\label{fig:baseExample}
\end{figure}

\paragraph{Contribution.}
In this paper, we aim to extend the results by Klawitter and Mchedlidze~\cite{KM21} to the case of three and more slopes.
We study the problem of deciding whether a digraph admits an upward planar $k$-slope drawing for any given $k$
with a special focus on the case $k = 3$.
Broadly speaking, we show how the problem becomes harder with more complex graph classes
and hence consider the following classes.
First, we use the classes of ordered and unordered directed trees as an easy introduction to the problem (\cref{sec:trees}). 
In particular, the upward planar slope numbers of these trees is easy to determine  
and upward planar $k$-slope drawings can be constructed efficiently.
Second, we show that for a given cactus digraph $G$ (defined below)
we can construct an upward planar $k$-slope drawing in polynomial time (\cref{sec:cactus}).
To this end, we devise a dynamic program on the block-cut tree of $G$ 
and utilize a simple polygon drawing algorithm by Culberson and Rawlins~\cite{CR85}.
Third, we show that it is NP-hard to decide whether
a given upward outerplanar
digraph admits an upward planar 3-slope drawing (\cref{sec:np}).
We extend the NP-hardness to $k > 3$ but restrict the graph class to upward planar
(except for $k = 4$ if no embedding is given).
Prior to that, we introduce further notation and terminology in the next section. 

Our findings and some known results are summarized in \cref{tab:overview}.

\begin{table}[t]
	\centering
	\begin{tabular}{@{}l | cccc | cccc@{}}
		Graph class & \multicolumn{4}{c|}{fixed embedding} & \multicolumn{4}{c}{variable embedding} \\
		\multicolumn{1}{r|}{$k =$}  & 2 & 3 & 4 & $\ge 5$ 				  & 2 & 3 & 4 & $\ge 5$ \\
		\midrule
		Tree        & P \cite{KM21} & P C\ref{clm:tree:fixed} & P C\ref{clm:tree:fixed} & P C\ref{clm:tree:fixed} 		& P T\ref{clm:tree:variable} & P T\ref{clm:tree:variable} & P T\ref{clm:tree:variable} & P T\ref{clm:tree:variable}  \\
		Cactus      & P \cite{KM21} & P T\ref{clm:cacti:fixed} & P T\ref{clm:cacti:fixed} & FPT T\ref{clm:cacti:fixed} 		  & P T\ref{clm:cacti:fixed} & P T\ref{clm:cacti:fixed} & P T\ref{clm:cacti:fixed} & FPT T\ref{clm:cacti:fixed} \\
		Outerplanar	& P \cite{KM21} & NPh T\ref{clm:nphard3slopesFixedEmb} & ? & ? 		  		& FPT \cite{KM21} & NPh C\ref{clm:nphard3slopesVariableEmb} & ? & ?  \\
		Planar	& P \cite{KM21} & NPh T\ref{clm:nphard3slopesFixedEmb} & NPh C\ref{clm:nphardkslopesFixedEmb} & NPh C\ref{clm:nphardkslopesFixedEmb}		  & FPT \cite{KM21} & NPh C\ref{clm:nphard3slopesVariableEmb} & ? & NPh T\ref{clm:nphardkslopesVariableEmb} \\
	\end{tabular}
	\caption{Complexity of constructing an upward planar $k$-slope drawing
		of a given digraph of the specified graph classes.
		P means polynomial-time solvable,
		FPT means fixed-parameter tractable, and
		NPh means NP-hard.
		Our findings contain a reference to a theorem (T)
		or a corollary (C).
	}
	\label{tab:overview}
\end{table}

\section{Notation and Terminology}
\label{sec:prelim}
All graphs in this paper are simple, that is,
they do not contain parallel edges or self-loops. 
Note that a given digraph $G$ may admit an upward planar $k$-slope drawing
only if it has maximum in- and out-degree at most $k$.
Otherwise, there would not be enough slopes to represent all edges entering or leaving a vertex
with too high in- or outdegree, respectively.
Hence, we assume from now on that every input digraph has maximum in- and out-degree at most $k$. 

Given a graph $G$, a planar embedding can be specified by the order of the edges around each vertex
as well as a designed outer face.
We speak of the \emph{fixed embedding scenario}, 
when besides~$G$, we are also given a planar embedding (or, here, an upward planar embedding) as input.
This embedding is then to be used also in the output drawing.
In the \emph{variable embedding scenario}, on the other hand,
we are not given an embedding and have to find one prior or parallel to constructing a drawing.

\paragraph{Graph Classes.}
A \emph{directed tree} is a digraph whose underlying graph is a tree,
that is, a connected graph that does not contain a cycle.
Note that directed trees are a superset of rooted trees,
that is, a directed tree can have multiple sources (vertices with indegree 0)
while a rooted tree has exactly one.
While for general graphs the terms fixed and variable embedding are common,
the following terms are used for trees:
an \emph{unordered} tree comes without planar embedding,
while for an \emph{ordered} tree, a planar embedding is specified.  

A \emph{cactus graph} is a planar graph where any two (simple) cycles share at most one vertex.
For digraphs, we assume that the underlying graph is a cactus graph.

An \emph{outerplanar graph} is a planar graph that admits an embedding where each vertex lies on the outer face.
A drawing of an outerplanar graph is usually expected to maintain this property.
We also assume this henceforth.
For the directed case,
we define an \emph{upward outerplanar digraph} as a digraph that admits
an upward planar drawing with each vertex on the outer face. 

Note that every directed tree is a cactus digraph
and that every cactus digraph is an outerplanar digraph. 
All of these three classes are subclasses of planar digraphs.

\paragraph{Slope Sets.}
When constructing an upward planar $k$-slope drawing,
we need a set of $k$ distinct slopes.
To obtain this slope set, we propose the following three settings, 
which are illustrated in \cref{fig:baseExample:grid,fig:baseExample:uniform,fig:baseExample:general}, respectively.
\begin{description}
\item[general setting] Any set of $k$ distinct slopes can be chosen.
\item[uniform angles setting] The slopes are distributed equally,
that is, clockwise from the x-axis they have angles in $\set{i \cdot \pi / k - \frac{\pi}{2k} \mid i \in \set{1, \dots, k}}$.
\item[regular grid setting] Any set of $k$ distinct slopes that connect points on the 2D grid
is allowed. This may include the horizontal slope and then
horizontally rightwards also counts as upwards.
We prefer to pick slopes that can be used by segments of roughly equal length; 
see for example \cref{fig:baseExample:grid} and \cref{fig:treeGrid}.
\end{description}

Both the uniform angles and the regular grid setting have their own advantages and disadvantages.
Uniform angles naturally lead to more balanced drawings
with more rotational symmetry, which we find more visually appealing.
Moreover, the drawings have a perfect angular resolution, which is also known as a quality measure of graph drawings~\cite{DETT99,DLM18}. 
The downside of the uniform angles setting is that we cannot always use grid points of the regular 2D grid.
For example, for $k = 3$, the first slope is $\tan (\pi/6) = 1/\sqrt{3}$,
which is an irrational number.
Henceforth, we assume for uniform angles a computational and representation model 
that can handle implicit coordinates or alternatively uses real numbers.
On the other hand,
while the regular grid setting naturally facilitates integer coordinates,
it may also yield less balanced angles between edges and irrational edge lengths. 
Since all of these settings have their natural justification,
we do not restrict our considerations to one of them.

Note that a 2-slope drawing can be sheared
such that one slope changes and the other remains the same; see \cref{fig:affinelyTransform3Slopes}ii;
we refer also to Klawitter and Mchedlidze~\cite{KM21}.
Moreover, observe that such a transformation only changes the length of segments with the changing slope,
but not the others.
The fact that this does not hold for three or more slopes 
introduces interesting new geometric aspects for these cases.
However, note that $k = 3$ is a special case
because no matter which three slopes we pick, they
can be affinely transformed to the slopes of the angles $\set{45^{\circ}, 90^{\circ}, 135^{\circ}} =  \set{\Sgiad, \Svert, \Sdiag}$,
as illustrated in \cref{fig:affinelyTransform3Slopes}.
Hence, when referring to $k = 3$, we may assume the use of this slope set.
(For illustrative purposes, we rotate drawings in \cref{sec:np} by $45^{\circ}$ clockwise and
then use the slope set $\set{\Svert, \Sdiag, \Shori}$.)
In some cases we number $k$ different slopes with the numbers 1 to $k$ in counterclockwise order.

\begin{figure}[t]
	\centering
	\includegraphics[page=1]{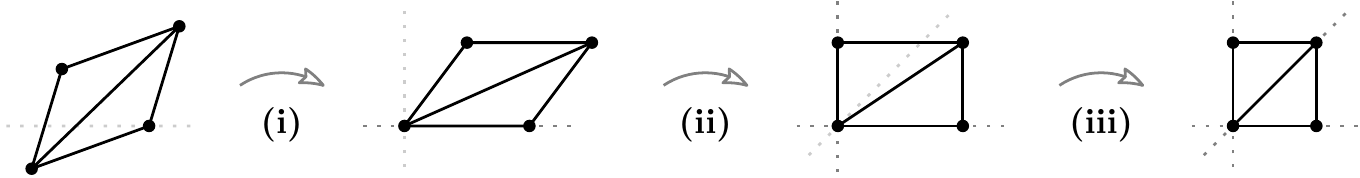}
	\caption{Given a drawing using a set of any three slopes, we can
		(i)~rotate,
		(ii)~shear, and
		(iii)~stretch
		it to a drawing with the slope set $\set{\nwarrow, \Svert, \Sdiag}$ (rotated here by $45^{\circ}$ clockwise for illustrative purposes to $\set{\Svert, \Sdiag, \Shori}$).}
	\label{fig:affinelyTransform3Slopes}
\end{figure} 

\paragraph{Slope Assignment.}
A \emph{$k$-slope assignment} of a digraph $G$ assigns each edge of $G$ one of $k$ slopes.
If~$G$ is upward plane, we call a $k$-slope assignment of $G$ \emph{consistent} 
if the assignment complies with the cyclic edge order around each vertex;
that is, for $k = 3$, if a vertex has three incoming edges,
they need to be assigned the slopes \Sdiag, \Svert, and \Sgiad in counterclockwise order. 
Clearly, if an upward plane embedding does not admit a consistent $k$-slope assignment,
then it also does not admit an upward planar $k$-slope drawing.

\section{Trees}
\label{sec:trees}
In this section, we consider upward planar $k$-slope drawings of directed trees. 
While our trees are in general not rooted, results for rooted trees can be derived
or are partially already known~\cite{BM11,BM13}.
Note that naturally every unordered tree is upward planar,
while an ordered tree is upward planar if and only if its embedding is bimodal. 

\subsection{Unordered Trees}
The planar slope number of an unordered undirected tree with maximum degree $\Delta$
is $\ceil{ \Delta / 2}$~\cite{DESW07}.
Therefore not surprisingly, the upward planar slope number of an unordered
directed tree $T$ equals the maximum in- and outdegree of~$T$.
To show this, we draw $T$ as subgraph of a larger, regular tree~$T_{k,h}$ for $h \geq 1$
where every non-leaf vertex has in- and outdegree $k$ and each leaf has distance~$h$ to a central vertex.
To draw $T_{k,h}$ on a grid with $k$ slopes, 
we adopt the strategy of Bachmaier, Brandenburg, Brunner, Hofmeier, Matzeder, and Unfried~\cite{BBBHMU08}
for complete rooted trees; see \cref{fig:treeGrid}.
Alternatively, $T_{k,h}$ can be drawn with $k$ uniform angles; see \cref{fig:treeNice}.

\begin{theorem} \label{clm:tree:variable}
Let $T$ be an unordered directed tree on $n$ vertices with maximum indegree and outdegree at most $k$.
Then $T$ admits an upward planar $k$-slope drawing in the regular grid setting
and $T$ admits an upward planar $k$-slope drawing in the uniform angles setting.
Moreover, the drawings can be computed in $\Oh(n)$ time.
\end{theorem}
\begin{proof}
Let $\ell$ be the number of vertices on a longest undirected path in $T$.
We first describe how to construct an upward planar $k$-slope drawing of $T' = T_{k, {\ceil{\ell / 2}}}$.
Let $\rho$ be the central vertex of $T'$.
Since $T$ is a subgraph of $T'$,
a drawing for $T$ could be obtained from the drawing of $T'$ straightforwardly.
However, since $T'$ might be substantially larger than $T$,
we describe at the end how to construct the drawing of $T$ directly.   

To draw $T_{k,{\ceil{\ell / 2}}}$ on a grid we first place $\rho$
at the center of an axis-aligned square that functions as the drawing region.
We then partition this larger square into $(\ceil{k / 2} + 1)^2$ smaller equal-sized squares.
The neighbors of $\rho$ are placed at the centers of the smaller squares 
that appear along the perimeter of the outer square.
If $k$ is even, then all smaller squares get occupied
where the left central square gets a predecessor of $\rho$
and the right central square gets a successor of $\rho$.
If $k$ is odd, then the left and right central squares stay unoccupied.
For each neighbor $v$ of $\rho$, we then proceed recursively within the square of $v$; see \cref{fig:treeGrid}.
Hence, the trees obtained from removing $\rho$ in $T'$ are all drawn in disjoint regions
and are thus non-overlapping.
The recursive procedure ends at the leaves and with the smallest squares.
The size of these smallest squares determines the final grid~size.
\begin{figure}[tb]
  \centering
  	\begin{subfigure}[t]{0.45 \linewidth}
		\centering
		\includegraphics[page=1]{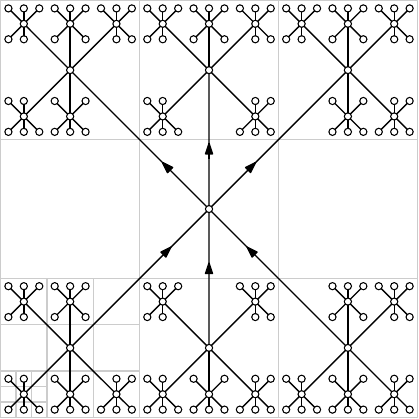}
		\caption{$T_{3,3}$, $k = 3$.}
		\label{fig:treeGrid:t33}
	\end{subfigure}
	\begin{subfigure}[t]{0.45 \linewidth} 
		\centering
		\includegraphics[page=2]{treeConstruction}
		\caption{$T_{6,2}$, $k = 6$.}
		\label{fig:treeGrid:t62}
	\end{subfigure}
  \caption{Upward planar $k$-slope drawings of unordered trees on the grid.}
  \label{fig:treeGrid} 
\end{figure}

Alternatively, to draw $T'$ with $k$ slopes and uniform angles,
we use a regular $2k$-gon (instead of a square) as drawing region and 
place $\rho$ at its center.
For the subtrees, we then use appropriately smaller regular $2k$-gons; see \cref{fig:treeNice}.
More precisely, we pick the sizes such that 
one corner of each smaller $2k$-gon is incident with a corner of the larger $2k$-gon,
and, for odd $k$, two small $2k$-gons touch at a corner, as in \cref{fig:treeNice:t33},
while for even $k$, they touch at a side, as in \cref{fig:treeNice:t62}.
The process continues recursively as before.
\begin{figure}[tb]
  \centering
  	\begin{subfigure}[t]{0.45 \linewidth}
		\centering
		\includegraphics[page=3]{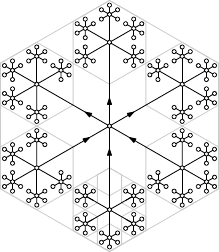}
		\caption{Unordered tree $T_{3,3}$, $k = 3$.}
		\label{fig:treeNice:t33}
	\end{subfigure}
	\begin{subfigure}[t]{0.45 \linewidth} 
		\centering
		\includegraphics[page=4]{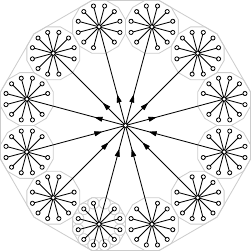}
		\caption{Unordered tree $T_{6,2}$, $k = 6$.}
		\label{fig:treeNice:t62}
	\end{subfigure}
  \caption{Upward planar $k$-slope drawings of unordered trees with uniform angles.}
  \label{fig:treeNice} 
\end{figure}

Lastly, we describe how to compute such drawings for $T$ directly.
First, find a longest path $P$ in~$T$ in linear time~\cite{BSZVGF02}
and determine a vertex $\rho$ in the middle of $P$.
Second, recursively construct a consistent $k$-slope assignment with $\rho$ as start vertex. 
For the next step, note that the squares or $2k$-gons 
are only used to find (and prove) appropriate edge lengths for each depth. 
However, we can also compute these lengths (and the grid size) directly in $\Oh(n)$ time.
Finally, again recursively from $\rho$, 
draw the edges of $T$ with their assigned slopes and the computed edge lengths.
Overall, each step and thus the whole algorithm runs in $\Oh(n)$ time. 
\end{proof}

Note that the recursive drawing procedure from \cref{clm:tree:variable} 
requires an exponential-size drawing area (or at least an exponential edge-length ratio).
For an arbitrary number of slopes, 
Frati~\cite{Fra08} showed that an area of size $\Theta(n \log n)$ suffices to draw an $n$-vertex tree upward planar.
As far as we know, it remains open whether
an upward planar $k$-slope drawing of an unordered directed tree with maximum in- and outdegree $k$ 
requires an exponential area.
This is the case for ordered directed trees, which we consider next.

\subsection{Ordered Trees}
Let $T$ be an ordered directed tree
that admits an upward planar drawing.
Hence, the vertices of $T$ are bimodal.
To determine the upward planar slope number of $T$,
it suffices to find a consistent $k$-slope assignment for $T$ with minimal~$k$.
We can then use the drawing algorithm from unordered trees.
In this regard, note that the maximum in- and outdegree are natural lower bounds
but that the choice of the (minimal) slope for an edge $vw$ cannot be determined locally at $v$ and $w$.
For example, the edge $vw$ in \cref{fig:treeSlopes:problem} is the third incoming edge at $w$
but requires at least slope 4,
since its preceding edge $uw$ already requires slope $3$ at $u$.
This effect only appears along alternating intervals of incoming and outgoing edges.
Hence we have the following observation; see also \cref{fig:treeSlopes:straight,fig:treeSlopes:alternating}.

\begin{observation}
The upward planar slope number of ordered directed trees with $n$ vertices, $n \geq 2$,
is bounded within $1$ and $n-1$ and these bounds are tight. 
\end{observation}

\begin{figure}[tb]
  \centering
  	\begin{subfigure}[t]{0.30 \linewidth}
		\centering
		\includegraphics[page=2]{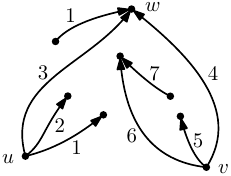}
		\caption{The third incoming edge $vw$ of $w$ cannot have slope 3.}
		\label{fig:treeSlopes:problem}
	\end{subfigure}
	$\quad$
	\begin{subfigure}[t]{0.30 \linewidth}
		\centering
		\includegraphics[page=3]{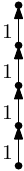}
		\caption{A directed path with upward planar slope number 1.}
		\label{fig:treeSlopes:straight}
	\end{subfigure}
	$\quad$
	\begin{subfigure}[t]{0.30 \linewidth}
		\centering
		\includegraphics[page=4]{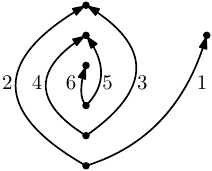}
		\caption{A directed path with upward planar slope number $n-1$.}
		\label{fig:treeSlopes:alternating}
	\end{subfigure}
  \caption{The upward planar slope number of an $n$-vertex ordered directed tree lies between 1 and $n-1$,
  since it is not determined locally.
  Here the number represents different slopes.}
  \label{fig:treeSlopes}
\end{figure}

A consistent $k$-slope assignment for $T$ minimizing $k$ can be constructed
with a simple greedy algorithm in linear time.

\begin{theorem} \label{clm:tree:slopenumber}
The upward planar slope number $k$ of an ordered directed tree on $n$ vertices
can be determined in $\Oh(n)$ time.
Moreover, an upward planar $k$-slope drawing of $T$ can be constructed in $\Oh(n)$ time.
\end{theorem}
\begin{proof}
Let $T$ be an ordered directed tree on $n$ vertices.
We describe a greedy algorithm 
that computes a consistent $k$-slope assignment for $T$ where $k$ is minimal.
The idea is that we assign the slopes $1, 2, \ldots$ in counterclockwise (ccw) order
around each vertex, but an edge $uv$ receives a slope
only if all outgoing edges at $u$ that ccw precede $uv$ 
and all incoming edges at $v$ that ccw precede $uv$ already got a sloped assigned.
Therefore, an edge $uv$ receives a mark when it has this property at $u$ or at $v$,
and gets added to a queue as soon as it receives its second mark. 
To find the first edge in the queue, we mark the ccw first incoming and outgoing edge at each vertex;
this can be done in overall $\Oh(n)$ time.
Observe that after this initialization,
there always exists an edge with two marks in a tree.

For an edge $uv$ that we take from the head of the queue,
we assign $uv$ the lowest available slope, 
that is, one plus the maximum of the slope of its ccw preceding outgoing edge at $u$ and its ccw preceding incoming edge at $v$;
if there is no preceding edge, then slope 1 is available. 
We then mark the ccw succeeding edges of $uv$ at $u$ and $v$.
Since $T$ is a tree, there is (until the algorithm terminates) 
always at least one edge in the queue. 

When the algorithm terminates, each edge is assigned a slope,
which increase in ccw order for the incoming and for the outgoing edges of each vertex.
Furthermore, the highest assigned slope is 
also the upward planar slope number of $T$ with respect to its embedding.
The running time of the algorithm is in $\Oh(n)$
since the algorithm considers each edge only a constant number of times.

Once we have computed the upward planar slope number $k$ of $T$ with respect to its embedding,
we can use the drawing algorithm from \cref{clm:tree:variable} 
to compute an upward planar $k$-slope drawing of $T$ again in $\Oh(n)$ time. 
\end{proof}

\begin{corollary} \label{clm:tree:fixed}
Let $T$ be an ordered directed tree on $n$ vertices with maximum indegree and outdegree at most~$k$.
We can decide in $\Oh(n)$ time whether $T$ admits an upward planar $k$-slope drawing.
\end{corollary}

Frati~\cite{Fra08} showed that an alternating ordered directed path on $n$-vertices as in \cref{fig:treeSlopes:alternating}
requires $\Omega(2^n)$ area. However, this path also requires $n-1$ slopes.
It is hence natural to ask whether there exist also ordered directed trees
that require an exponential area in upward planar $k$-slopes drawings for a constant $k$.
This question has recently been answered by Quapil~\cite{Qua21} for $k = 3$
by giving a similar path with a spiral structure
that requires exponential area;
the example can be extended to higher $k$ and, by adding a single triangle somewhere, to cactus digraphs.

\section{Cactus Graphs}
\label{sec:cactus}
\begin{figure}[tb]
	\centering
	\includegraphics{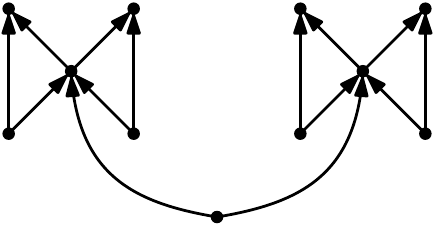}
	\caption{Cactus digraph with maximum in- and outdegree 3,
		which admits even in the variable embedding scenario 
		no upward planar 3-slope drawing as shown by Quapil~\cite{Qua21}.
	}
	\label{fig:cactus:no3slopes}
\end{figure}
In this section, we show that
deciding whether a given cactus digraph~$G$
admits an upward planar~$k$-slope drawing
is fixed-parameter tractable (FPT) in $k$.
This holds for both, the fixed and the variable embedding scenario.
In particular, constructing an upward planar $k$-slope drawing
of a cactus digraph can be done in polynomial time if $k$ is constant.

Roughly speaking, we use a dynamic program on the block-cut tree of~$G$
that computes combinable~$k$-slope assignments for each block.
Note that an acyclic cactus digraph~$G$ is always upward planar
and if no upward planar embedding is specified,
then the algorithm can compute one.
However, as Quapil~\cite[Figure 4.3]{Qua21} pointed out, 
unlike directed trees, not every cactus digraph with maximum in- and outdegree at most~$k$
admits an upward planar~$k$-slope drawing;
see \cref{fig:cactus:no3slopes}.

\begin{figure}[tb]
  \centering
	\begin{subfigure}[t]{0.31 \linewidth}
		\centering
		\includegraphics[page=3]{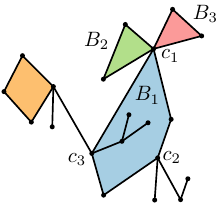}
		\caption{A cactus digraph $G$.}
		\label{fig:cactus:example:G}
	\end{subfigure}
	\begin{subfigure}[t]{0.31 \linewidth}
		\centering
		\includegraphics[page=4]{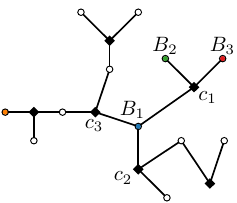}
		\caption{Block-cut tree of $G$.}
		\label{fig:cactus:example:T}
	\end{subfigure}
	\begin{subfigure}[t]{0.31 \linewidth}
		\centering
		\includegraphics[page=5]{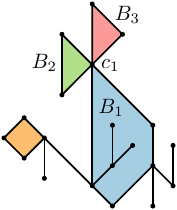}
		\caption{3-slope drawing of $G$.}
		\label{fig:cactus:example:D}
	\end{subfigure}
  \caption{The vertices (blocks) of the block-cut tree of a cactus digraph~$G$ 
  correspond to the cycles and edges of~$G$; we draw all blocks separately and then merge their drawings.}
  \label{fig:cactus:example}
\end{figure}

Recall that a block-cut tree~$\cT$ of a graph~$G$ has a vertex 
for each \emph{block} (biconnected component) and each cut vertex of~$G$
and an edge between a block~$B$ and a cut vertex~$c$ if~$c$ is part of~$B$;
see \cref{fig:cactus:example:T}. 
Let~$G$ be a cactus digraph. Note that in a block-cut tree~$\cT$ of~$G$
each block vertex is either a cycle or an edge 
-- we thus distinguish between \emph{cycle blocks} and \emph{edge blocks}.
The block-cut tree of~$G$ can be computed in linear time~\cite{Tar72}.
For a cactus digraph~$G$ to admit an upward planar~$k$-slope drawing,
each block of~$\cT$ must be drawable under constraints imposed by other blocks.
Consider for an example the cactus~$G$ in \cref{fig:cactus:example:G,fig:cactus:example:T,fig:cactus:example:D}
under the given embedding.
For~$k = 3$,
the two edges of the block~$B_1$ incident to the cut vertex~$c_1$
need the slopes \Svert and \Sgiad because of the blocks~$B_2$ and~$B_3$. 
Our strategy is thus as follows.

\paragraph{Algorithm.}
In the first phase, we run a dynamic program (described below) on the blocks of~$\cT$
to find a consistent slope assignment for each block
such that the blocks are {\em combinatorially} combinable.
If successful, we enter the second phase, 
where we compute drawings of the blocks that are {\em geometrically} combinable.
In the last phase we put all block drawings together.

Let~$G$ be a cactus with blocks~$B_1, \ldots, B_\ell$ 
and let~$\cT$ be the block-cut tree of~$G$.
We pick an arbitrary block vertex, say~$B'$, of~$\cT$ as root and direct all edges towards~$B'$.
As a result, each block vertex~$B$ (except~$B'$) has exactly one outgoing edge towards a cut vertex~$c$ in~$\cT$.
We then say~$c$ is the \emph{anchor} of~$B$.
Let~$B$ be a cycle block with anchor~$c$ 
and suppose we have a slope assignment for~$B$.
Let~$e$ and~$e'$ be the edges of~$B$ incident to~$c$
such that clockwise (cw) along the inner face of~$B$ we have the sequence~$e, c, e'$.
Then the \emph{anchor type}~$t_c(B)$ of~$c$ for~$B$
is defined as the slopes of~$e$ and~$e'$ and if they are incoming or outgoing edges at~$c$; see \cref{fig:cactus:categories}.
For an edge block~$B$ with edge~$e$, the \emph{anchor type}~$t_c(B)$ describes the slope of~$e$ 
and if~$e$ is incoming or outgoing at~$c$.
For cycle blocks and edge blocks, there are~$2k \cdot (2k-1)$ and~$2k$ different anchor types, respectively.

\begin{figure}[tb]
	\centering
	\includegraphics{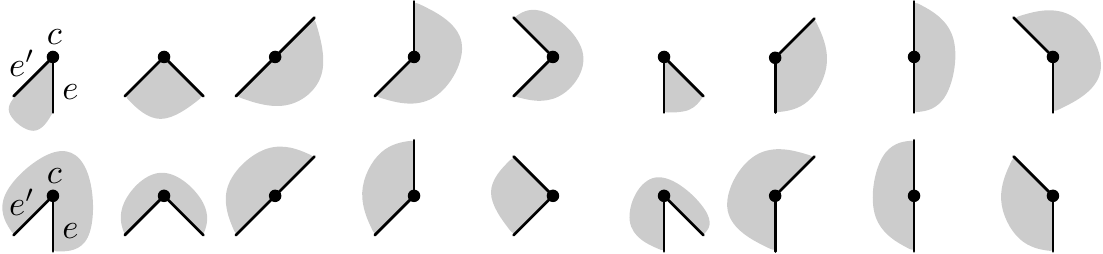}
	\caption{A subset of the anchor types of a cycle block for~$k = 3$.}
	\label{fig:cactus:categories}
\end{figure}

For a block vertex~$B$ with anchor~$c$, 
a \emph{block tuple}~$\tau_B = \langle \phi_B, t_c(B) \rangle$
consists of a consistent~$k$-slope assignment~$\phi_B$ of~$B$ 
and an anchor type~$t_c(B)$ of~$c$.
A block tuple~$\tau_B$ is \emph{feasible}
if $B$ has no descendant blocks or if $B$'s descendant blocks admit
a non-empty set of feasible block tuples that can be combined with $\tau_B$.
A \emph{feasible set} for~$B$ is a maximal set of feasible block tuples for~$B$
that have pairwise different anchor types.
We process~$\cT$ in a post-order traversal.
For each block we compute a feasible set based on the feasible sets
of its descendant blocks.

\paragraph{Combinatorial Realization.}

Computing the feasible set of a cycle block~$B$ with anchor~$c$ works as follows.
Let~$B$ be the cycle~$(c = v_1, e_1, v_2, e_2, \ldots, v_{\abs{B}}, e_{\abs{B}}, v_1)$ 
-- if an embedding is given, let this order be cw around the inner face.
Roughly speaking, in a dynamic programming approach, 
we find for each edge $e_i$, $i \in \set{1, \ldots, \abs{B}}$, 
all possible slopes of~$e_i$ based on the possible slopes of~$e_{i-1}$.
Furthermore, we also consider how far we have rotated along the boundary of $B$ from $v_1$ to $v_{i+1}$
to keep track of whether we have a total sum of inner (outer) angles around $B$ of $\pm 2\pi$.
More precisely, 
we walk around~$B$ once and store for each edge $e_i$, $i \in \set{1, \ldots, \abs{B}}$, 
all tuples of (i) a possible slope $s$ of $e_i$ and 
(ii) the corresponding sum~$\alpha$ of rotation angles when traversing $B$ from $v_1$ to $v_{i+1}$ (consistently, along one side of $B$).
We call the tuple $\croc{s, \alpha}$
an \emph{edge tuple} of~$e_i$.
As a base for the sum of rotation angles, when considering $e_1$,
we use the angle a horizontal ray emanating at $v_1$ and pointing towards positive infinity
would need to rotate to contain $e_1$; see \cref{fig:cactus:cycleWalk:a}.

Next we describe how to compute all edge tuples of the edges of a block~$B$.
We start with computing the edge tuples of~$e_1$ 
by considering all combinations of feasible block tuples of descendant blocks of~$B$ anchored at~$v_1$ and~$v_2$.
For each possible combination, we can test which slopes $e_1$ may get in $\Oh(k)$ time. 
The corresponding rotation angle can be computed in $\Oh(1)$ time.
A descendant cycle block can have $\Oh(k^2)$ many feasible block tuples,
while a descendant edge block can have $\Oh(k)$ many feasible block tuples.
Together at $v_1$ and $v_2$, there can be at most $2(k-1)$ descendant cycle blocks
and at most $4(k-1)$ descendant edges.
The total number of possible combinations of feasible block tuples of descendant blocks of $B$
is thus at most  $\Oh((k^2)^{2k-2}) = \Oh(k^{4k-4})$.
(To see this, note that pairs of edge block descendants 
can be considered as cycle block descendants for this calculation.) 
Therefore, we can find all possible edge tuples of $e_1$ in $\Oh(k^{4k-3})$ time. 
In the example of \cref{fig:cactus:cycleWalk:a} for~$k = 3$, 
assuming a fixed embedding, the edge~$e_1$ can only have slope \Svert 
and we have thus rotated~$90^\circ$.
For this edge tuple~$\croc{\Svert, 90^\circ}$, the edge~$e_2$ in \cref{fig:cactus:cycleWalk:b}
has also only one possible slope, namely~\Sdiag, 
and the rotation increases by~$135^\circ$ to a total of~$225^\circ$.
However, in the variable embedding scenario,~$e_1$ can also have slopes \Sdiag and \Sgiad, see \cref{fig:cactus:cycleWalk:c,fig:cactus:cycleWalk:d}.

In general, to compute the edge tuples of an edge~$e_i$,~$i \in \set{2, \ldots, \abs{B}}$,
we can use the same procedure to determine all possible slopes of~$e_i$
as we have described for~$e_1$.
The main difference is that we need to take into account also all edge tuples of~$e_{i-1}$.
We can treat the slope of an edge tuple of~$e_{i-1}$
like another descendant edge block.
Additionally we update the sum of rotation angles for $e_i$.
The number of different edge tuples of~$e_{i-1}$ is in $\Oh(k \abs{B})$
because each edge tuple of $e_{i-1}$ has one of $k$ slopes,
the sum of angles for each pair of edge tuples having the same slope
differs by a multiple of~$2 \pi$, and
each of the $\Oh(\abs{B})$ previous edges adds a rotation of $< 2 \pi$.
Thus, since we have $\Oh(k \abs{B})$ edge tuples for $e_{i-1}$,
$\Oh(k^{4k-4})$ feasible block tuples of descendant blocks,
and $\Oh(k)$ possible slopes for~$e_i$,
we can compute all edge tuples of~$e_i$ in $\Oh(k^{4k-2} \abs{B})$ time.
For all edges of $B$ this takes $\Oh(k^{4k-2}{\abs{B}}^2)$ time.

Note that we may obtain the same edge tuple in multiple ways~--
the actual number of edge tuples for each~$e_i$ is again $\Oh(k \abs{B})$.
For each edge tuple~$t$ of~$e_i$, we store a pointer to $t$ at the $\Oh(k)$ edge tuple(s) of~$e_{i-1}$ that $t$ is based on.
In other words, we build a digraph $H$ on the edge tuples of the edges $e_1, \ldots, e_{\abs{B}}$.
The digraph~$H$ has $\Oh(k \abs{B}^2)$ vertices and $\Oh(k^2 \abs{B}^2)$ edges.

\begin{figure}[tb]
  \centering
  	\begin{subfigure}[t]{0.23 \linewidth}
		\centering
		\includegraphics[page=3]{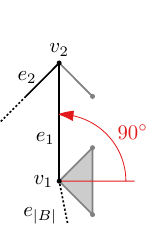}
		\caption{For a fixed embedding, $e_1$ starts with slope~\Svert, \ldots}
		\label{fig:cactus:cycleWalk:a}
	\end{subfigure}$\;\;$
	\begin{subfigure}[t]{0.22 \linewidth}
		\centering
		\includegraphics[page=4]{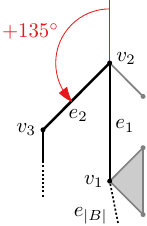}
		\caption{\ldots and then $e_2$ gets slope~\Sdiag.}
		\label{fig:cactus:cycleWalk:b}
	\end{subfigure}$\;\;$ 	 
	\begin{subfigure}[t]{0.23 \linewidth}
		\centering
		\includegraphics[page=5]{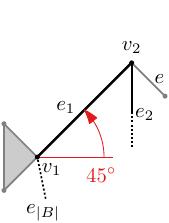}
		\caption{In the variable embedding scenario,
		$e_1$ may get slope \Sdiag \ldots}
		\label{fig:cactus:cycleWalk:c}
	\end{subfigure}$\;\;$
	\begin{subfigure}[t]{0.24 \linewidth}
		\centering
		\includegraphics[page=6]{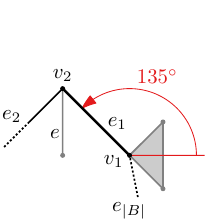}
		\caption{\ldots or $e_1$ may get slope~\Sgiad.}
		\label{fig:cactus:cycleWalk:d}
	\end{subfigure}
  \caption{Computing a slope assignment for a cycle block $B$ with anchor $v_1$.
  	The algorithm handles the edges of $B$ one by one starting with $e_1$.
  }
  \label{fig:cactus:cycleWalk}
\end{figure}

When we handle~$e_{\abs{B}}$,
we discard all tuples that do not result in a~$2\pi$ rotation if the embedding is given
or with~$\pm 2\pi$ if no embedding is given. 
This ensures that the cycle has a geometric realization~\cite{CR85}.
To determine the feasible set for~$B$,
we start a breadth-first search (BFS) on $H$ for each edge tuple of $e_1$.
For each edge tuple of $e_{\abs{B}}$ we reach,
we combine the slope of~$e_1$ and~$e_{\abs{B}}$ as well as whether
the rotation is~$+2\pi$ or~$-2\pi$ to get an anchor type of~$B$ at~$c$. 
By backtracking from the edge tuple of~$e_{\abs{B}}$ within the BFS tree,
we find a consistent slope assignment of~$B$, which yields a feasible block tuple.

We have $\Oh(k)$ BFS runs, which take $\Oh(k^2 \abs{B}^2)$ time each
and $\Oh(k^3 \abs{B}^2)$ time in total.
Backtracking for all $\Oh(k^2)$ anchor types is in $\Oh(k^2 \abs{B})$ time in total.
Hence, these steps are dominated by the computation of the edge tuples
and we can compute a feasible set for~$B$ in $\Oh(k^{4k-2}{\abs{B}}^2)$ time.
Over all blocks of~$G$, this in $\Oh(k^{4k}n^2)$
since each vertex is in at most $k$ blocks
and, therefore, we have $\sum_{i = 1}^\ell \abs{ B_i }^2 \leq (kn)^2$.

Finally, let us describe how to obtain a combinatorial realization for~$G$.
If the feasible set of the root block~$B'$ is non-empty,
then there is a combinatorial realization of~$G$.
To find one, we pick any feasible block tuple of~$B'$
and try to combine it with the feasible block tuples of its descendant blocks.
(We know that a consistent combination exists.)
Since each feasible set has size $\Oh(k^2)$, we can find
a consistent set of feasible block tuples of descendant blocks
in $\Oh((k^{2})^{k-1})$ time per vertex of~$B'$.
Over all vertices of~$G$ this is in $\Oh(k^{2k-2} n)$ time.

\paragraph{Geometric Realization.}
Suppose we have found a combinatorial realization in the form of
a consistent $k$-slope assignment for every cycle and edge of~$G$.
In the variable embedding scenario, we now know whether and how cycles nest.
We thus re-root $\cT$ such that the root block lies on the outer face.
In the following, we describe how to obtain a drawing of a cycle block $B$ as a polygon
that does not intersect the edges of its parent block $B'$ at its anchor point~$c$.

We describe this only for the uniform angles setting
and leave it as an open question for the regular grid setting.
Given any sequence~$\sigma$ of rational angles (i.e., a rational number times $\pi$)
that sum up to $\pm 2 \pi$, Culberson and Rawlins~\cite{CR85} 
describe an algorithm that outputs a polygon with~$\sigma$ as turning angles.
Their so-called \textit{Turtlegon} algorithm works as follows.
It defines a base angle $\alpha$ as the greatest common divisor of $\pi$ and all angles in~$\sigma$;
in our case this is $\pi / k$
(w.l.o.g.\ we can assume that our drawings are slightly rotated by $\frac{\pi}{2k}$).
Larger angles are split into sequences of~$\pm \alpha$ resulting in a new angle sequence~$\sigma'$.
W.l.o.g.\ let $\sigma'$ contain more angles $+\alpha$ than $- \alpha$.
Using some of the $\alpha$s,
the Turtlegon algorithm draws a regular $(2 \pi/\alpha)$-gon (in our case $2k$-gon).
To accommodate additional angles in between, it inserts exponentially
shrinking detours at the corners of the $(2 \pi/\alpha)$-gon; see \cref{fig:cactus:boundingboxes}b.
In the end, we get the original larger angles from merging the smaller angles~\cite{CR85}; see \cref{fig:cactus:boundingboxes}c.

The difficulty for us when employing this $\Oh(k\abs{B})$ time algorithm, 
is to ensure that the edges of the parent block $B'$ 
can reach the anchor point~$c$ without intersecting the polygon of $B$.
This might be impossible if $c$ lies within a spiral inside a detour.
However, we can avoid this if we let an incident edge of~$c$
be a side of the $2k$-gon (this is always possible because we can pick
an appropriate set of $\alpha$ angles of $\sigma'$ for the $2k$-gon)
and if we let each detour edge shrink by a
sufficiently large factor (e.g., $k{\abs{B}}$);
see \cref{fig:cactus:boundingboxes}.

The running time of this step is in $\Oh(k{\abs{B}})$.
Since each vertex is in at most $k$ blocks, 
we have that $\sum_{i = 1}^\ell \abs{ B_i } \leq kn$ and,
hence, the total running time is in $\Oh(k^2n)$. 

\begin{figure}[tb] 
  \centering
 	\begin{subfigure}[t]{.31 \linewidth}
 		\centering
 		\includegraphics[page=3]{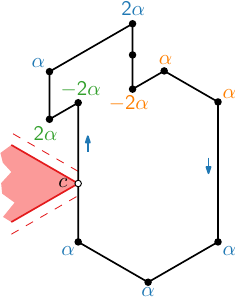}
 		\caption{A combinatorial description of a block $B$ with angle sequence that contains large angles \ldots}
 	\end{subfigure}
 	\hfill
 	\begin{subfigure}[t]{.31 \linewidth}
 		\centering
 		\includegraphics[page=4]{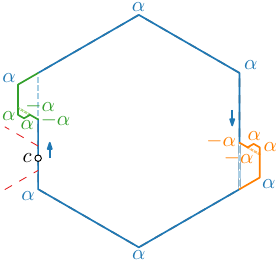}
 		\caption{\ldots is converted to an angle sequence on base angle $\alpha$
 		and realized as $2k$-gon with detours.}
 	\end{subfigure}
 	\hfill
 	\begin{subfigure}[t]{.31 \linewidth}
 		\centering
 		\includegraphics[page=5]{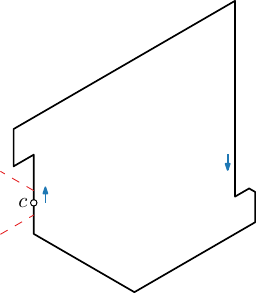}
 		\caption{Small angles originating from large angles are merged back together.}
 	\end{subfigure}
  \caption{The Turtlegon algorithm by Culberson and Rawlins~\cite{CR85}
  realizes an angle sequence of a block $B$ with anchor point $c$ by first splitting large angles,
  constructing a $2k$-gon with exponentially shrinking detours,
  and then merging small angles back together. When using the Turtlegon algorithm, 
  we ensure that the anchor point~$c$ lies at a $2k$-gon edge; here $k = 3$.}
  \label{fig:cactus:boundingboxes}
\end{figure}

\paragraph{Putting Blocks Together.}
We start with a drawing of the root block.
We then recursively draw each child (in a BFS-like order) such that its anchor 
point coincides with the corresponding vertex of the parent polygon
and scale down the drawing of the child block
such that the appended polygon does not intersect the existing drawing.
Note that it always suffices to scale down each child to the size
of the minimum distance of any two vertices within in the parent polygon.
We can determine vertex pairs of minimum and maximum distance
for a block~$B$ in $\Oh({\abs{B}} \log {\abs{B}})$ time
and then place and scale each polygon in linear time.

The total running time is dominated by the dynamic program
and thus in $\Oh(k^{4k}n^2)$.
For a constant $k$, this is a quadratic time algorithm.

\begin{theorem} \label{clm:cacti:fixed}
Let $G$ be an upward planar (or plane) cactus digraph with maximum in- and outdegree at most~$k$.
It can be constructively tested in $\Oh(k^{4k}n^2)$ time
whether $G$ admits an upward planar $k$-slope drawing
in the uniform angles setting.
In other words, constructing an upward planar $k$-slope drawing of a
cactus digraph is fixed-parameter tractable in~$k$.
\end{theorem}

For the regular grid setting, we
cannot use the algorithm by Culberson and Rawlins~\cite{CR85}
because we have irrational multiples of $\pi$ as turning angles.
For a sequence of general turning angles, the algorithm by Hartley~\cite{Har89}
computes a polygon realizing that sequence.
However, it is not immediately clear how to guarantee that the
edges of the parent polygon at the anchor point are not intersected.
For general polygons, we believe that we can iteratively shrink the spikes
to resolve potential intersections.
Since such a procedure involves some more technicalities,
we leave it as an open question.

\section{Outerplanar and Planar Graphs}
\label{sec:np}
In this section, we show that for any constant $k \ge 3$ deciding 
whether an upward planar digraph admits an upward planar $k$-slope drawing is NP-hard.
For $k = 3$, we show NP-hardness for this decision problem
even for upward outerplanar digraphs.
Except for $k = 4$, this hardness holds true regardless of whether we prescribe an embedding or not.

However, it remains open if the problem is also NP-complete.
Containment in NP is not immediately clear since it is open whether some digraphs
require irrational (or super-polynomial precise) coordinates for any $k$-slope drawing.
More precisely, if the problem was in NP, there would be small proof certificates for yes-instances 
that a verifier could use to decide the problem in polynomial time.
Typically a combinatorial characterization or a drawing of the input graph could act as such a certificate.
However in our case, we do not know whether there are digraphs 
 that require irrational (or super-polynomial precise) coordinates 
and if so, how to treat them implicitly
or, alternatively, how a combinatorial characterization would look like.

We first describe our NP-hardness reduction for embedded outerplanar digraphs for $3$ slopes.
Remember that for an arbitrary number of slopes, 
upward planarity for outerplanar digraphs can be decided in polynomial time~\cite{Pap94}~-- 
if an embedding is given, even in linear time.
Afterwards, we show how this NP-hardness reduction can be extended to the variable embedding scenario and, then, to larger~$k$ while using planar
instead of outerplanar digraphs.

In this section,
we consider drawings rotated by 45$^\circ$
(upwards is to the top right) using the
regular grid slope set $\set{\Svert, \Sdiag, \Shori}$.
This rotation facilitates the visualization of large orthogonal structures
in our figures.

\subsection{Outerplane Digraphs and \texorpdfstring{$k = 3$}{k = 3}}

In this section, we show that it is NP-hard to decide
whether a given upward outerplanar digraph with a given upward outerplanar embedding
admits an upward outerplanar drawing using only three distinct slopes
by reduction from \pms.

\pms (sometimes also known as \pmrs) is an NP-complete version of \textsc{3-SAT}~\cite{BK12},
where the three literals of each clause are all either
negated or unnegated~-- from now on called \emph{negative} and \emph{positive} clauses, respectively.
Moreover, the incidence graph\footnote{
	The incidence graph of a \textsc{SAT} formula has a vertex for each variable and clause, and 
	an edge for each occurrence of a variable in a clause between the corresponding vertices.}
has a planar drawing where 
the vertices are rectangles, the edges are vertical straight-line segments,
the variables are arranged on a horizontal line,
the positive clauses are above, and the negative clauses are below this line;
see \cref{fig:pm3sat-instance}.
For our reduction, we can allow that a clause contains
the same literal multiple times.
Hence, we can also assume that all clauses contain
exactly three literals (as otherwise we can fill up
clauses by duplicating literals and obtain an equivalent formula).

For a given \pms
formula~$F$ and a rectangular drawing of its incidence graph, 
we construct a corresponding upward outerplanar digraph~$G_F$,
which can only be drawn upward planar with 3 slopes if $F$ is satisfiable.
Our construction follows ideas of Nöllenburg~\cite{Noe10} and Kraus~\cite{Kra20},
and utilizes the fact that some digraphs can only be drawn in a way that they admit helpful properties arising in the upward planar drawing setting specifically for $k=3$.
We use them as building blocks in our reduction as subgraphs of $G_F$.
Let us describe these special (sub)graphs next.

We start with the digraph \gsquare
with vertices $s$, $x$, $y$, and $t$, and
edges $st$, $su$, $ut$, $sv$, and $vt$,  as depicted  in \cref{fig:gsquaregraph}.
Observe that, up to scaling and mirroring along a diagonal axis,
\gsquare admits an upward planar 3-slope drawing only as an outerplanar square as in \cref{fig:gsquaredrawn}.
We can attach multiple squares (and triangles) to each other as in \cref{fig:unitsquares}.
The drawing of such a bigger digraph is unique up to scaling and mirroring diagonally.
If the squares form a tree, the drawing is outerplanar.
We refer to these squares as \emph{unit squares}, since, once set, the side lengths for all attached squares are the same.

We also use \gsquares to construct our next building block~-- the digraph \gshift.
To allow a certain small degree of freedom in our construction,
we exploit the following property of~\gshift.

\begin{figure}[t]
	\centering
	\begin{subfigure}[t]{0.16 \linewidth}
		\centering
		\includegraphics[trim=0 20 350 0,clip]{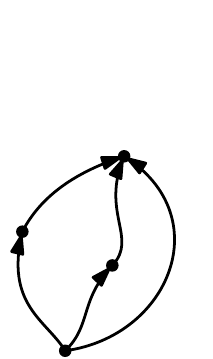}
		\caption{Digraph \gsquare.}
		\label{fig:gsquaregraph}
	\end{subfigure}
	\hfill
	\begin{subfigure}[t]{0.17 \linewidth}
		\centering
		\includegraphics[trim=70 20 275 0,clip]{baseSquareAndSlider}
		\caption{Unique 3-slope realization of \gsquare.}
		\label{fig:gsquaredrawn}
	\end{subfigure}
	\hfill
	\begin{subfigure}[t]{0.18 \linewidth}
		\centering
		\includegraphics[trim=150 20 190 36,clip]{baseSquareAndSlider}
		\caption{Combination of \gsquares and triangles.}
		\label{fig:unitsquares}
	\end{subfigure}
	\hfill
	\begin{subfigure}[t]{0.4 \linewidth}
		\centering
		\includegraphics[trim=256 20 0 0,clip]{baseSquareAndSlider}
		\caption{Digraph \gshift.}
		\label{fig:gshift}
	\end{subfigure}
	\caption{Building blocks used as subgraphs in our NP-hardness reduction.
		The digraph \gsquare admits only an upward 3-slope drawing as square.
		By combining copies of \gsquare and triangles, we can build larger rigid structures.
		Then, all \gsquares have the same size and we refer to them as unit squares.
		The digraph \gshift admits only upward 3-slope drawings with one
		degree of freedom.
	}
	\label{fig:baseSquareAndSlider}
\end{figure}

\begin{lemma} \label{clm:Gshift}
	In any upward planar 3-slope drawing of \gshift (see \cref{fig:gshift})
	\begin{itemize}
		\item the edges $e_1$ and $e_2$ are parallel and have the same arbitrary length $\ell > 0$,
		\item all edges are oriented as in \cref{fig:gshift} up to mirroring along a diagonal axis, and
		\item all vertical and horizontal edges (excluding $e_1$ and $e_2$) have the same lengths, as well as all diagonal edges.
	\end{itemize}
\end{lemma}
\begin{proof}
As subgraphs, we have twice a pair of attached \gsquares.
As observed before, each pair can only be drawn as attached unit squares.
Without loss of generality, let the first two squares (containing vertices $v_2$, $v_3$, $v_5$) 
be drawn vertically above each other as in \cref{fig:gshift}.
Since $v_2$ has outdegree 3 and $e_3 = v_2v_5$ has slope~\Svert, we can argue that the edge~$e_1$ gets slope \Shori.
If~$e_1$ used \Sdiag, the edge~$e_5$ would have no slope to close the triangle~$(e_3, e_4, e_5)$ because, 
at~$v_5$, the rightmost outgoing slope is only~\Shori.
Thus, $e_4$ gets~\Sdiag and the only remaining slope for $e_1$ is \Shori.
Note that $e_3$ is a unit square edge and the adjacent edges $e_4$ and $e_5$ form a triangle with $e_3$
that has the same size as half a unit square. 
Hence, $e_4$ ($e_5$) has the same length as all diagonal (resp.\ horizontal) edges.

The same argument as for the outgoing edges of $v_2$ applies to the incoming edges of~$v_4$.
However, assume for contradiction, 
that the whole block on the right side of~\cref{fig:gshift} is mirrored diagonally, 
i.e., $v_1$ and $v_4$ are on the same horizontal line and $e_2$ gets slope~\Svert.
Then, $v_4$, $v_1$, and $v_2$ would be on the same horizontal line.
This contradicts $v_3$ being vertically below its neighbors $v_2$ and $v_4$.
Hence, the right block has the same orientation as the left one, and $e_2$ gets slope \Shori and is parallel to $e_1$.
This implies that the distances between $v_2$ and $v_3$ on the one hand, 
and $v_1$ and $v_4$ on the other hand are the same 
and the \gsquares (and their attached triangles) have the size and orientation as in \cref{fig:gshift}.
By the orientation of $v_1, v_2, v_3, v_4$, the edges $e_1$ and $e_2$ have an identical length~$\ell$.
Only a value $\ell \leq 0$ would cause a non-planar drawing.
\end{proof}

\begin{figure}[t]
	\centering
	
	\begin{subfigure}[t]{1 \linewidth}
		\centering
		\includegraphics[width=.54\linewidth,page=1]{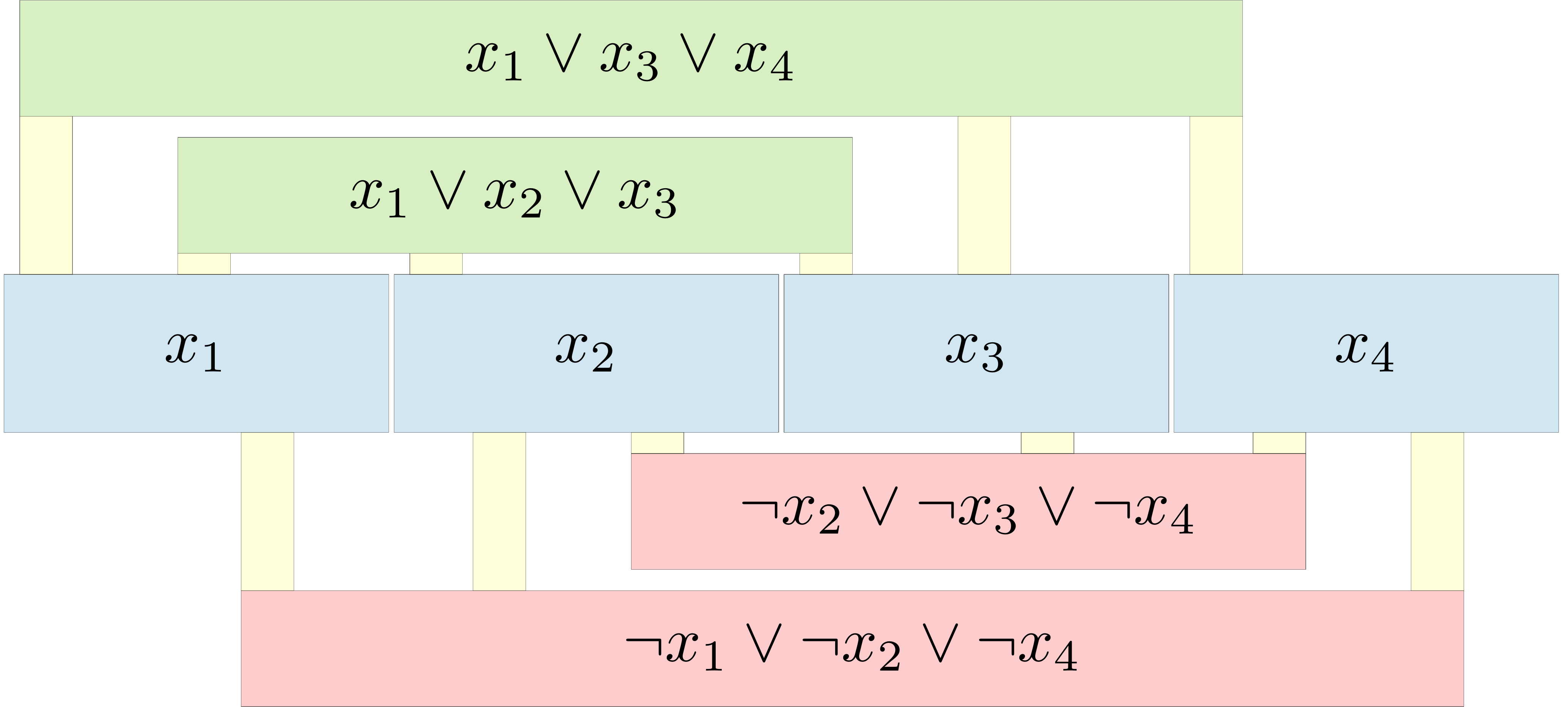}
		\caption{Rectangular incidence graph drawing of a \pms formula~$F$.}
		\label{fig:pm3sat-instance}
	\end{subfigure}
	
	\bigskip
	
	\begin{subfigure}[t]{1 \linewidth}
		\centering
		\includegraphics[width=1\linewidth,page=3]{construction}
		\caption{Outerplanar drawing of the digraph~$G_F$ obtained from~(a).
			Chains of unit squares are drawn as straight-line segments.
			The variable/clause/edge gadgets occupy the areas of their corresponding rectangles.
			Here, $x_1$ and $x_4$ are set to \textsf{false} ({\re brush} on the left side within their variable gadgets),
			while $x_2$ and $x_3$ are set to \textsf{true} ({\re brush} on the right side within their variable gadgets)}
		\label{fig:construction-3slopes}
	\end{subfigure}
	\caption{Schematic example for our NP-hardness reduction.}
	\label{fig:construction}
\end{figure}

With this construction kit of useful (sub)graphs in hand, we build a digraph
whose upward planar drawings represent the satisfying truth assignments for~$F$.

The \textbf{high-level construction} is depicted in \cref{fig:construction-3slopes}.
We construct, for each variable~$x_i$, a specific digraph~-- the \emph{variable gadget for $x_i$} (blue in \cref{fig:construction-3slopes}).
Similarly, for each clause~$c_j$, there is a specific digraph~-- the \emph{clause gadget for~$c_j$} (green and red in \cref{fig:construction-3slopes}).
All gadgets mainly consist of chains of \gsquares.
For a drawing, this enforces a rigid frame structure built from unit squares.
We glue all variable gadgets together in a row and connect variable and clause gadgets by \emph{edge gadgets} (yellow in \cref{fig:construction-3slopes})
such that the composite digraph remains upward outerplanar
and all \gsquares are drawn as unit~squares.

\begin{figure}[t]
	\centering
	\includegraphics[width=.92\textwidth]{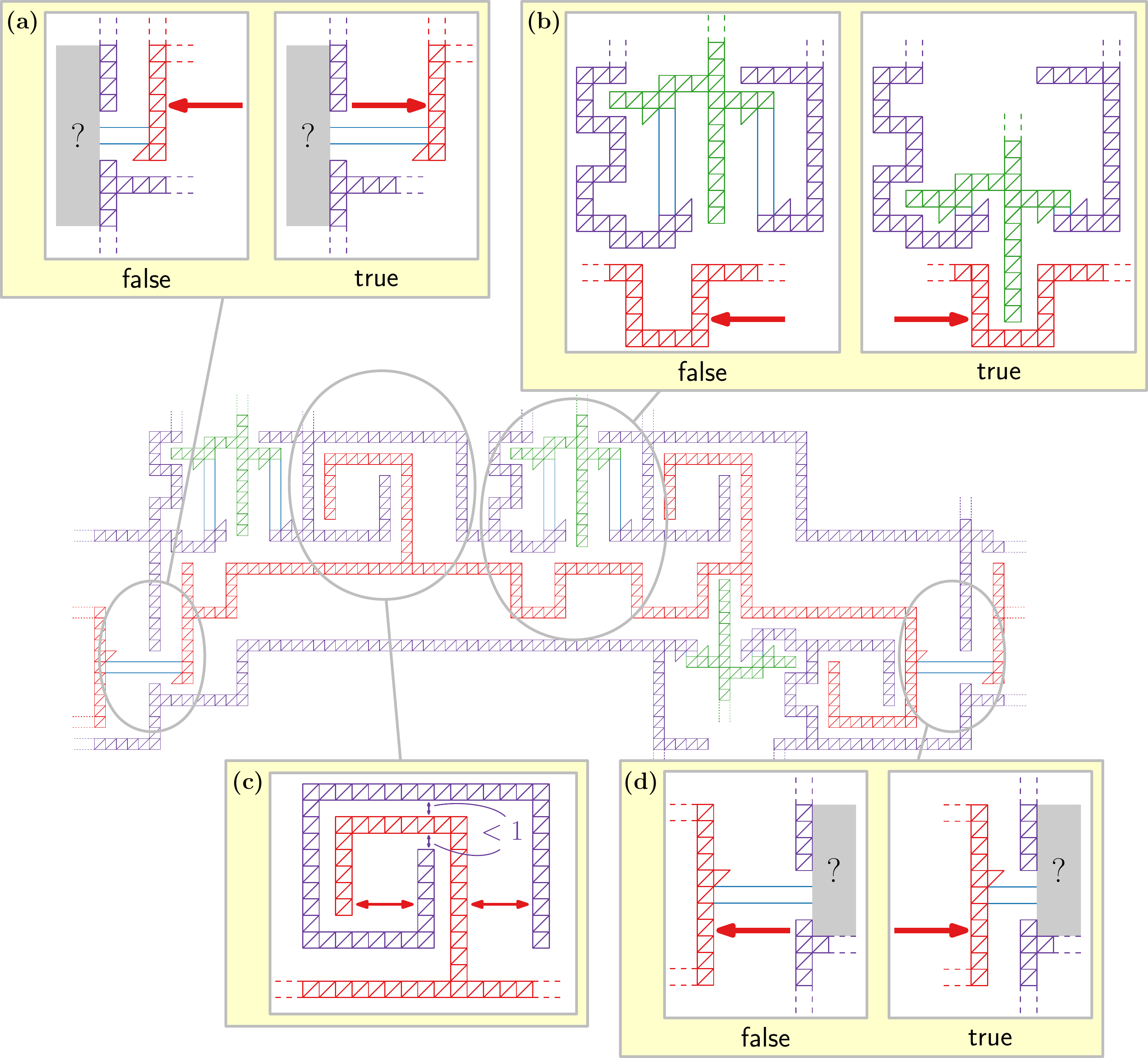}
	\caption{A variable gadget, which is contained in two positive and one negative clauses.
		The {\re brush} is positioned to the left and, thus, the variable is set to \textsf{false}.}
	\label{fig:variableGadget}
\end{figure}

\begin{figure}[t]
 	\centering
 	\begin{subfigure}[t]{.31 \linewidth}
 		\centering
 		\includegraphics[page=4,width=1.0\textwidth]{clauseGadget}
 	\end{subfigure}
 	\hfill
 	\begin{subfigure}[t]{.31 \linewidth}
 		\centering
 		\includegraphics[page=6,width=1.0\textwidth]{clauseGadget}
 	\end{subfigure}
 	\hfill
 	\begin{subfigure}[t]{.31 \linewidth}
 		\centering
 		\includegraphics[page=7,width=1.0\textwidth]{clauseGadget}
 	\end{subfigure}
 	
 	\smallskip
 	
 	\begin{subfigure}[t]{.31 \linewidth}
 		\centering
 		\includegraphics[page=5,width=1.0\textwidth]{clauseGadget}
 	\end{subfigure}
 	\hfill
 	\begin{subfigure}[t]{.31 \linewidth}
 		\centering
 		\includegraphics[page=3,width=1.0\textwidth]{clauseGadget}
 	\end{subfigure}
 	\hfill
 	\begin{subfigure}[t]{.31 \linewidth}
 		\centering
 		\includegraphics[page=2,width=1.0\textwidth]{clauseGadget}
 	\end{subfigure}
 	
 	\smallskip
 	
 	\begin{subfigure}[t]{.31 \linewidth}
 		\centering
 		\includegraphics[page=1,width=1.0\textwidth]{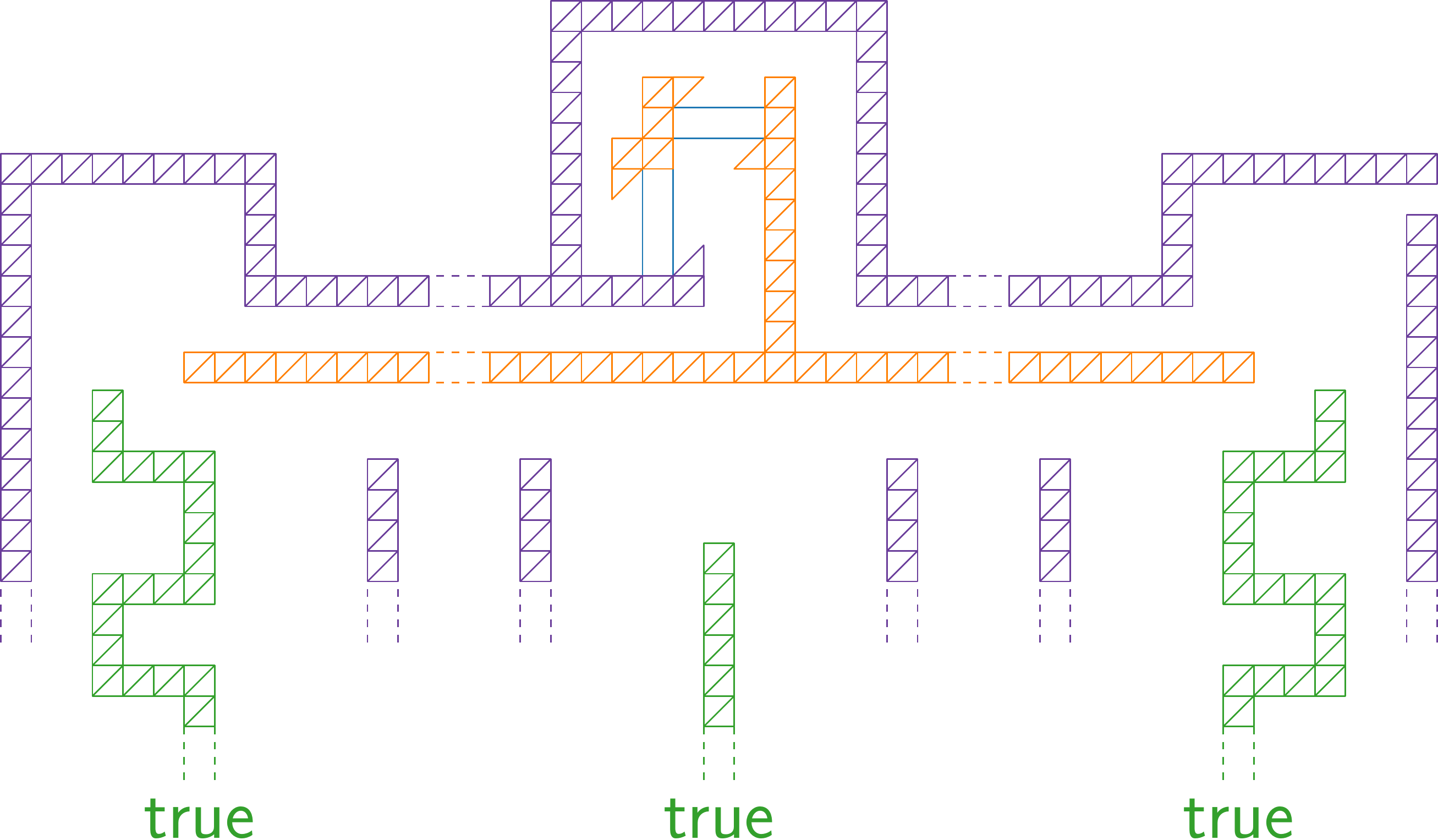}
 	\end{subfigure}
 	\hfill
 	\begin{subfigure}[t]{.65 \linewidth}
 		\centering
 		\includegraphics[page=8,width=1.0\textwidth,trim=0 70 0 0,clip]{clauseGadget}
 	\end{subfigure}
 	
	\caption{Positive clause gadget
	in 8 configurations.}
 	\label{fig:clauseGadget}
\end{figure}

A \textbf{variable gadget} is depicted in \cref{fig:variableGadget}.
Its base structure is the {\pu (violet) frame} composed of chains of unit squares.
The core element is the {\re (red) central chain} of unit squares (with a few {\re side-arms}),
which has one degree of flexibility, namely, moving as a whole to the left or to the right without leaving the {\pu frame structure} of the gadget.
It looks and behaves a bit like a pipe cleaning brush that is stuck inside the {\pu frame} but can be moved a bit back and forth.
Hence, we call it a \emph{\re brush}.
It is connected via a \gshift to the {\re brush} of the previous variable gadget (see \cref{fig:variableGadget}a/d) and
the first {\re brush} is connected to the {\pu frame} via a \gshift
(see on the left side of \cref{fig:construction-3slopes}).
This allows only a horizontal shift of the {\re brushes}, but no vertical movement relative to its anchor point at the {\pu frame structure}.
Note that the horizontal position in any variable gadget is independent of those in all other gadgets. 
If the {\re brush} is positioned to the very left (right), the corresponding variable is set to \textsf{false} (\textsf{true}). 

For each occurrence of a variable in a positive clause, we have a construction as depicted in \cref{fig:variableGadget}b.
There, a long chain of {\gr (green) \gsquares} -- from now on called \emph{\gr bolt} -- is attached to the {\pu frame structure} via two~\gshifts,
which allow only a vertical, but no horizontal shift.
The {\gr bolt} has on its left side an {\gr arm}, which can only be placed in one of two {\pu pockets of the frame}.
It can always be placed in the {\pu upper pocket}, which pushes the {\gr bolt} outwards with respect to the variable gadget (into an edge gadget and a clause gadget).
It can only be placed in the {\pu lower pocket} if the {\re brush} is shifted to the very right (i.e. set to \textsf{true})~-- then the {\gr bolt} can ``fall'' into a cove of the {\re brush}.
For each occurrence of a variable in a negative clause, we have this construction
upside-down, such that the {\gr bolt} can be pulled into the variable gadget only if the {\re brush} is shifted to the very left (i.e. set to \textsf{false}).

Note that, to maintain outerplanarity of the whole construction, the {\pu frame structure} is not contiguous, but connected by \gshifts
and the {\gr arms} of the {\gr bolts}.
Hence, the {\pu frame structure} decomposes into many components that have fixed relative horizontal positions
and their unit squares have the same side lengths. 
However, the components can shift up and down relative to each other.
To keep this vertical shift small enough not to affect the correct functioning of our reduction,
we use, for each such component, the construction depicted in \cref{fig:variableGadget}c.
The chain of {\re brushes} has no vertical flexibility and serves as a base ground for an ``anchor'' of the {\pu frame}.
The {\pu frame} can move less than one unit up or down unless it violates planarity.
If the {\pu frame} would be shifted up enough to be completely above the {\re brush},
it would get in conflict with the adjacent {\gr bolt}. 
 
An \textbf{edge gadget} consists of only three straight chains~-- two {\pu frame segments} and a {\gr bolt} in the middle.
Their purpose is to synchronize the distance of the clause gadgets to the variable gadgets and to preserve the size of the unit squares.
Several edge gadgets are depicted on yellow background color in \cref{fig:construction-3slopes}.

A \textbf{clause gadget} for a positive clause is depicted in \cref{fig:clauseGadget}.
Within a {\pu frame}, which is connected at six points to the {\pu frames} of three edge gadgets, 
there is a horizontal {\og (orange) bar}, which is attached via two \gshifts to the {\pu frame}~-- 
one \gshift allows a horizontal, the other allows a vertical shift.
It resembles a crane that can move up and extend its arm, while it holds the horizontal {\og bar} on a vertical {\og (orange) rope}.
The three {\gr bolts} from the corresponding variable gadgets reach into the clause gadget.
The lengths of these {\gr bolts} is chosen such that, if they are pushed out of their variable gadget and into the clause gadget, they only slightly fit inside the gadget.
Depending on whether each of the {\gr bolts} is pushed into the clause gadget or pulled out of it,
we have eight possible configurations (with sufficiently small vertical slack).
They represent the eight possible truth assignments to a clause.
In \cref{fig:clauseGadget}, we illustrate that in each configuration, we can accommodate the horizontal {\og bar} in an upward planar 3-slope drawing of the clause gadget 
-- except for the case when all three {\gr bolts} push into the clause gadget, which represents the truth assignment \textsf{false} to all contained variables.

A negative clause gadget uses the same construction, but mirrored vertically.
There, three {\gr bolts} pushing into the clause gadget means the contained variables are all set to \textsf{true}.

Putting our gadgets together, we conclude:

\begin{theorem} \label{clm:nphard3slopesFixedEmb}
		Deciding whether an upward outerplane digraph admits an upward planar 3-slope drawing is NP-hard.
\end{theorem}
\begin{proof}
	To show NP-hardness, we use the reduction from \pms as described above.
	Let $F$ be a given \pms formula with a rectangular drawing of the incidence graph of~$F$ and $G_F$ be the digraph obtained by our reduction.
	Each gadget in~$G_F$ has only polynomial size and we can construct it in polynomial time.
	
	If $F$ is satisfiable, then there is a satisfying assignment~$T^\star$ of truth values to the variables of $F$.
	Draw $G_F$ as illustrated in \cref{fig:construction-3slopes}, \cref{fig:variableGadget}, and \cref{fig:clauseGadget} 
	such that the {\re brush} in a variable gadget has a distance of~$\varepsilon$ (for some sufficiently small $\varepsilon > 0$)
	to the {\pu frame} on its left if the corresponding variable is set to \textsf{false} in~$T^\star$
	and has a distance of~$\varepsilon$ to the {\pu frame} on its right if the corresponding variable is set to \textsf{true}.
	This drawing uses only three slopes and is upward planar.
	Observe that it is also outerplanar.
	On a chain of unit squares, the vertices on both sides are incident to the outer face.
	The drawings of the variable gadgets are open from one to the other
	and the drawing of the rightmost variable gadget is open to the outside.
	Moreover, the drawings of the positive (negative) clause gadgets are open on their top (bottom) right and on their whole bottom (top) sides, 
	which also gives the ``covered'' clause gadgets and the outer sides of the variable gadgets access to the outer face.
	
	On the other hand, if there is an upward planar 3-slope drawing of $G_F$ and it resembles the structure in~\cref{fig:construction-3slopes}, 
	we can read a satisfying truth assignment depending on the positions of the {\re brushes} 
	(for the ones in intermediate position where no {\gr bolt} can use the {\pu lower pocket}, we can use any assignment for this variable).
	
	In the remainder of this proof, we argue that any drawing resembles this structure.
	Consider the \gsquare of the {\pu frame structure} which is connected via a \gshift to the {\re brush} of the first variable gadget.
	Clearly, it is drawn as a unit square~-- it is our reference unit square having side length~1.
	Observe that all other unit squares are connected via chains of \gsquares or~ \gshifts to this reference unit square.
	By our observation and \cref{clm:Gshift}, they have the same side length.
	When ignoring the central parallel edges in the \gshifts, the drawing decomposes to few rigid components.
	The first rigid {\pu frame component} contains the reference unit square.
	The central edges of the incident \gshift connecting it to the first {\re brush} can have only relatively small length 
	as they would hit the ``back wall'' of, again, the first rigid {\pu frame component} otherwise.
	Hence, the {\re brush} is indeed drawn inside the {\pu frame} of the first variable gadget.
	Moreover, the first rigid {\pu frame component} is connected via a \gshift to a {\gr bolt}, 
	which in turn is connected to the next rigid {\pu frame components} of the first variable gadget (see \cref{fig:variableGadget}b).
	Since this {\gr bolt} cannot escape the upper boundary of the clause gadget, which is also part of the first rigid {\pu frame component},
	the {\gr arm} is in one of the two {\pu pockets}.
	Then, the start point of the next rigid {\pu frame structure} cannot be above the {\re arm} of the {\re brush} in \cref{fig:variableGadget}c 
	and the construction depicted there keeps the vertical slack of this rigid {\pu frame component} within the range $(-1, 1)$.
	This argument inductively propagates for all following rigid components.
	Hence, we have vertical slack of less than 1 for the {\pu frame structures} and
	for the {\gr arms} in the {\pu pockets}.
	In the configuration \textsf{false}-\textsf{false}-\textsf{false} of the clause gadget,
	this would give us $2 - \varepsilon$ vertical gain if two {\pu frames} move away from each other and $1 - \varepsilon$ gain if the {\gr arm} is close to the bottom of its {\pu pocket}.
	However, we would require a shift of the horizontal {\og bar} of more than 3 to be below or above the bulge of the left or the right {\gr bolt}.
	Hence, this configuration is not drawable and all other positions of the horizontal {\og bar} correspond to one of the 7 other, satisfying configurations of \cref{fig:clauseGadget}.
\end{proof}

\subsection{Outerplanar Digraphs and \texorpdfstring{$k = 3$}{k = 3}}
Note that in $G_F$, we have used only connected \gsquares and \gshifts.
By our observation on chains of unit squares and by \cref{clm:Gshift},
the planar embedding of $G_F$ is unique up to mirroring along a diagonal axis.
Therefore, our reduction holds true also for the variable embedding scenario and we derive the following corollary from \cref{clm:nphard3slopesFixedEmb}.

\begin{corollary} \label[corollary]{clm:nphard3slopesVariableEmb}
	Deciding whether an upward outerplanar digraph admits an upward planar 3-slope drawing is NP-hard
	in the variable embedding scenario.
\end{corollary}

\subsection{Plane Digraphs and \texorpdfstring{$k \geq 4$}{k equal or greater than 4}}
Next, we describe how to extend our NP-hardness reduction to more than $3$ slopes.
There, however, we use only upward planar instead of upward outerplanar digraphs.
Observe that, if we fix the embedding and give up outerplanarity, 
we can add dummy leaves to each vertex to occupy all but the originally used slopes.
Since any 3 slopes can be projected to $\{\Svert, \Sdiag, \Shori\}$
and we block all other slopes, our arguments work for all sets of $k$ slopes and the reduction remains correct.
The digraph remains upward planar, but it is not upward outerplanar any more.
We formalize this statement in the following corallary
derived from \cref{clm:nphard3slopesFixedEmb}.

\begin{corollary} \label{clm:nphardkslopesFixedEmb}
	Deciding whether an upward plane digraph
	(i.e., a bimodal embedding is given)
	with maximum in- and outdegree~$k$ admits an upward planar drawing with $k$ slopes is NP-hard for $k \ge 3$.
	This holds true for all choices of $k$ slopes.
\end{corollary}

\subsection{Planar Digraphs and \texorpdfstring{$k \geq 5$}{k equal or greater than 5}} 
Last, we consider the variable embedding scenario.
We remark that, given a digraph, finding an upward planar embedding
is already NP-hard~\cite{GT01} for all $k \ge 2$~\cite{KM21}.
This NP-hardness immediately propagates to our problem.
However, we can show that the problem remains NP-hard
even if we can find an embedding (or multiple different embeddings) efficiently.
Hence, NP-hardness additionally comes from finding a concrete drawing
when we are allowed to change any upward planar embedding arbitrarily%
\footnote{We may assume, we obtain an upward planar embedding
	to a given digraph by an oracle.}.

We show that our NP-hardness reduction remains applicable for $k \ge 5$
by extending our digraph~$G_F$.
We prove that for any upward planar embedding of~$G_F$,
deciding whether $G_F$ admits a $k$-slope drawing remains NP-hard.
This leaves $k = 4$ in the variable embedding scenario as the only open case.
More precisely, we extend $G_F$ such that it has
a unique planar embedding up to mirroring along a diagonal axis
and up to swapping subgraphs only used to occupy the slopes we do not
use for our original NP-hardness reduction with three slopes.

Assume for now that $k$ is an odd number;
thereafter, we consider the case that $k$ is even.
From the given $k$ slopes, we pick the 3 middle slopes to host the digraph of the hardness construction described before.
For simplicity, we visualize these 3 middle slopes again as $\{\Svert, \Sdiag, \Shori\}$ 
and the other slopes in quadrants II and IV around a vertex.
The key idea is to occupy the unused slopes at each vertex by
\emph{fans} and \emph{beaters} as depicted in \cref{fig:fan,fig:beater} instead of simple leaves.
Fans are appended to the outside of each vertex $v$ if the angle $\alpha$
that has been formed in the old construction is at least $180^{\circ}$.
More precisely, for each unoccupied slope spanned by $\alpha$, we add a neighbor to $v$
and then connect each consecutive pair of these new neighbors of $v$;
edge directions are set appropriately.
For each other remaining slope $s$ at each vertex~$v$, we add a beater.
This is a digraph obtained from the wheel graph $W_{2k+1}$ as follows.
The \emph{wheel graph} $W_{2k+1}$ is the cycle $C_{2k} = (v_1, \dots, v_{2k})$
with an additional vertex $c$ that is adjacent to
all other vertices $\set{v_1, \dots, v_{2k}}$.
The edges of $W_{2k+1}$ on the cycle of the wheel are directed from
a local source $v_{\lfloor k/2 \rfloor}$ to a local sink~$v_{k + \lfloor k/2 \rfloor}$,
the vertices $v_1, \dots, v_k$ have outgoing edges towards~$c$, 
and 
the vertices $v_{k+1}, \dots, v_{2k}$ have incoming edges from~$c$.
Furthermore, one spoke~$e^\star = (c, v_i)$ (for a particular $i \in \set{1, \dots, 2k}$)
is broken free and attached to~$v$,
that is, we remove the edges $v_{i-1}v_i$ and $v_iv_{i+1}$ (indices modulo $2k$)
and we identify $v_i$ as $v$; see \cref{fig:beater}.
This construction enforces an order on the spokes and, hence, we choose $i$ such that we can prescribe the slope $s$ of $e^\star$.  
Note that the whole beater could be mirrored leaving two possible slopes for $e^\star$. 
However, this is unproblematic since in our construction the ``mirrored'' slope of $s$
is also occupied by a beater or a fan.
For an illustration how to append fans and beater see \cref{fig:gadgetsKgreater3odd}.

Next, we prove that this suffices to enforce a desired embedding
and we describe how to extend the construction when $k$ is an even number.
Though we lost upward outerplanarity, 
note that the underlying undirected graph remains outerplanar for odd $k$.

\begin{figure}[t]
	\centering
	\begin{subfigure}[t]{0.45 \linewidth}
		\centering
		\includegraphics[trim=70 232 180 18,clip]{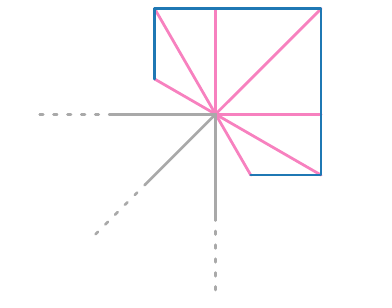}
		\caption{A fan.}
		\label{fig:fan}
	\end{subfigure}
	\hfill
	\begin{subfigure}[t]{0.45 \linewidth}
		\centering
		\includegraphics[trim=200 220 10 0,clip]{gadgetsKgreater3}
		\caption{A beater.}
		\label{fig:beater}
	\end{subfigure}

	\bigskip

	\begin{subfigure}[t]{1.0 \linewidth}
		\centering
		\includegraphics[trim=0 23 80 142,clip]{gadgetsKgreater3}
		\caption{We add fans and beaters to each vertex of the digraph such that all unused slopes are occupied.}
		\label{fig:gadgetsKgreater3odd}
	\end{subfigure}
	\caption{Example for $k = 5$ illustrating how to occupy unused slopes
		such that our NP-hardness reduction with three slopes remains applicable.
	}
	\label{fig:gadgetsKgreater3}
\end{figure}

\begin{theorem} \label{clm:nphardkslopesVariableEmb}
Deciding whether an upward planar digraph with maximum in- and outdegree~$k$ admits an upward planar drawing with $k$ slopes is NP-hard
for $k \in \mathbb{N}^+ \setminus \{1, 2, 4\}$
in the variable embedding scenario
even if we can find an upward planar embedding efficiently.
This holds true for all choices of $k$ slopes.
\end{theorem}
\begin{proof}
	For $k = 3$, we use \cref{clm:nphard3slopesVariableEmb}.
	Now assume that $k \ge 5$ is an odd number.
	We will consider the case where $k$ is even afterwards.
	We show next that the extended digraph can only be embedded in the
	previously described way up to mirroring along a diagonal axis
	and up to swapping beaters at a vertex.
	
	Observe that the {\pk (pink) edges} of each fan occupy neighboring slopes because their other endpoints are connected by {\bl (blue) edges}.
	For every fan, the slope set it covers contains both incoming and outgoing edges.
	Therefore, there are at most two ways to add a fan to a vertex~--
	either to its left or to its right side.
	In any case, it needs to be added to the outer face since otherwise
	neither a \gsquare nor a \gshift can be drawn (any fan blocks $\ge k -1$ slopes, which corresponds to $\ge 180^{\circ}$).
	
	The inner {\gr (green) edges} of a beater occupy all but one slope, which remains for the
	{\og (orange) edge}~$e^\star$ connecting the beater to its vertex.
	We next analyze the slope of~$e^\star$.
	Because of the outer {\bl edges} of a beater,
	there is an order of the slopes of the inner {\gr edges}
	(up to mirroring the whole beater)
	which determines the assignment of slopes within the beater.
	Hence, {\og edge} $e^\star$ uses one of two possible slopes~--
	say the $\ell$-th or the $(k - \ell + 1)$-th slope.
	Since $k$ is an odd number, the middle slope $\ell = \lceil k / 2 \rceil$ is unique, which fixes the position of some beaters
	(e.g., in \cref{fig:gadgetsKgreater3odd}, the beater in the outer face on the right side).
	For the other beaters, observe that we use the $(\lceil k / 2 \rceil - 1)$-th
	and the $(\lceil k / 2 \rceil + 1)$-th slope never for a beater, but always
	for the edges of the original construction or fans.
	For all other values of $\ell$, the opposite $(k - \ell + 1)$-th slope is
	always occupied by another beater or a~fan.
	
	This means, the embedding of the subgraph used for the reduction is the same as in the digraph~$G_F$ from the reduction of
	\cref{clm:nphard3slopesFixedEmb} up to mirroring,
	and the embedding of the whole digraph is unique up to mirroring and exchanging the positions of pairs of beaters.
	
	\begin{figure}[t]
		\centering
		\includegraphics[page=3, trim={79, 144, 0, 12} ,clip]{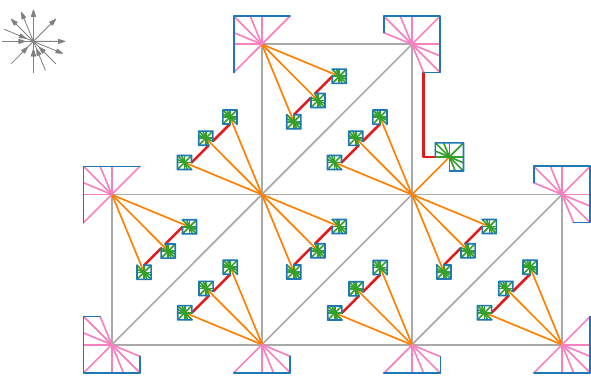}
		\caption{Example for $k = 6$ illustrating how to occupy unused slopes
			such that our NP-hardness reduction with three slopes remains applicable.
			We add {\re (red) dummy edges} to connect the beaters to each other or an adjacent fan, 
			which fixes the faces, in which the beaters end up,
			in the case $k$ is even and greater than 4.
		}
		\label{fig:gadgetsKgreater3even}
	\end{figure}
	
	It remains to consider the case when $k \ge 5$ is even.
	Our goal is to use for the digraph from \cref{clm:nphard3slopesFixedEmb} always the same 3 middle slopes, e.g., the $(k/2 - 1)$-th, $k / 2$-th, and $(k/2 + 1)$-th slope
	as in \cref{fig:gadgetsKgreater3even}.
	This is automatically the case if all fans and beaters are drawn in the desired face.
	We next describe how we can enforce in which face a beater or fan is drawn.
	At each vertex, we connect all incident beaters that shall share a common face
	by {\re (thick red) dummy edges} as depicted in \cref{fig:gadgetsKgreater3even}.
	If there is only one beater at a vertex at a face,
	we connect it by a {\re path} of length~2 to a neighboring fan; see the one on the top right of \cref{fig:gadgetsKgreater3even}.
	This latter case concerns only beaters on the outer face that block the slope \Sdiag.
	(Inside a \gshift, we do not need beaters in between $e_1$ and $e_2$.)
	In every other case, observe that within each bundle of connected beaters being adjacent to a vertex~$v$,
	there is at least one beater being connected to $v$ by an outgoing {\og edge},
	and there is at least one beater being connected to $v$ by an incoming {\og edge}.
	Hence, this bundle must be placed in the face being incident to $v$ where
	the direction of {\gy (gray) bounding edges} switches.
	This is either the opposite side of a fan or the opposite side of
	another bundle of beaters.
	In the latter case, the bundle of beaters cannot be exchanged since this would leave an incoming and an outgoing slope unused,
	which is not possible if $v$ has in- and outdegree $k$.
\end{proof}

\subsection{Planar Digraphs and \texorpdfstring{$k = 4$}{k = 4}}
\label{sec:planark4}
We have shown that it is NP-hard to decide whether
a given digraph admits an upward planar drawing with $k$ slopes
for all $k \ge 3$ in the fixed and the variable embedding scenario
except for $k = 4$ in the variable embedding scenario.
The reason that we could not show NP-hardness for this specific setting
is due to our extended hardness construction using beaters and fans.
We use beaters and fans to block all but three designated slopes
for the base hardness construction on three slopes.
However, beaters can be drawn in two different ways
(by mirroring) and hence potentially block two different slopes.
The slopes blocked by the same beaters can be grouped into pairs.
If we have an odd number of slopes, the central slope is
not part of such a pair (a mirrored beater still blocks the central slope).
We thus use the central slope and a pair of slopes 
for our three slopes in the base construction.

If $k$ is even, we connect beaters in order to block a sequence of slopes at a vertex~$v$~--
these are slopes for both incoming and outgoing edges incident to~$v$.
We exploit the property that at~$v$,
the incoming and outgoing edges form a
contiguous sequence in any upward embedding
and, hence, there are two turning points of edge directions at~$v$.
Choosing three neighboring slopes from an even number of available slopes
leaves two sets of unused neighboring slopes (then blocked by beaters)
whose cardinalities are off by one.
Hence, at a vertex with degree~$2k$, the number of incident incoming
and outgoing edges belonging to a group of connected beaters 
differs for the left and the right turning point.
For $k = 4$ we use all but one slope for the base construction.
This means, we cannot connect beaters and, thus, cannot cover
these turning points at the vertices.
Hence, a beater cannot reliably block the fourth slope 
but may use one of the three slopes of the base construction
and an edge of the base construction may then use the fourth slope.

\section{Concluding Remarks}
\label{sec:conclusion}
We have investigated the problem of
whether a digraph admits an upward planar $k$-slope drawing
for a fixed set of $k$ slopes.
Roughly speaking, the boundary between polynomial-time solvability
(for constant $k$)
and NP-hardness lies somewhere between cacti and planar digraphs.
In our analysis, we have somewhat sidestepped the issue of representing the coordinates of our drawings.
This is also the reason why we only showed NP-hardness,
but not NP-completeness for our problems.
Like some other geometric problems,
there is a chance that our problem is also $\exists \mathbb{R}$-complete.
(Remember that determining the (upward) planar slope number
is $\exists \mathbb{R}$-hard in general~\cite{Hof17,Qua21}.)

There remain several more open questions.
Let us start with the question marks in \cref{tab:overview}.
For planar digraphs, the only open case is $k = 4$
in the variable embedding scenario.
It would be good to close this gap~--
in \cref{sec:planark4}, we discuss why our current hardness construction fails.
For outerplanar digraphs and $k > 3$,
it is unclear whether deciding if there is
an upward outerplanar $k$-slope drawing is also NP-hard.
Supposed these cases were also NP-hard,
where does it become polynomial-time solvable?
Containment in P is not even clear for cacti and $k \in \omega(1)$,
where we could only give an FPT algorithm.
One may try to find an optimization to our dynamic program
to get rid of the dependence on $k$
in the exponent in the running time of the algorithm.
Furthermore, one may extend the algorithm to the regular grid setting
or to an arbitrary set of $k$ slopes.

In general, area consumption and area requirement is a question
worth investigating for many of our cases.
In particular, it remains open whether
an upward planar $k$-slope drawing of an unordered directed tree with maximum in- and outdegree $k$ 
requires sometimes exponential area.

The \emph{segment number} of a graph $G$ is the smallest number of line segments needed 
for a planar straight-line drawing of $G$~\cite{DESW07}; it serves as another quality measure of graph drawings,
which also indicates the use few geometric primitives.
As far as we know, the segment number of upward planar drawings has not been studied yet.
So overall, we hope to see more research on upward planar drawings that use few slopes
or that use few segments.

\pdfbookmark[1]{References}{References} 
\bibliographystyle{abbrvurl}
\bibliography{sources}

\end{document}